\documentclass[a4paper,11pt]{article}

\usepackage[letterpaper,  margin=1in]{geometry}
\usepackage{graphicx}

\usepackage{amsthm, amsmath,amsfonts,graphicx,hyperref,color,caption,subcaption,float,soul,multirow}
\usepackage{todonotes}
\usepackage[numbers,sort]{natbib}
\usepackage{authblk}
\usepackage[noend]{algpseudocode}
\usepackage{microtype,amssymb}
\usepackage{thm-restate}
\usepackage{enumerate} 
\usepackage{paralist}
\usepackage[utf8]{inputenc}
\usepackage[T1]{fontenc}
\usepackage{paralist}
\usetikzlibrary{decorations.pathreplacing}
\usetikzlibrary{calc}
\usetikzlibrary{shapes.misc}
\tikzset{cross/.style={cross out, draw=black, minimum size=2*(#1-\pgflinewidth), inner sep=0pt, outer sep=0pt},cross/.default={1pt}}
\usepackage{framed}
\usepackage[inline]{enumitem}

\newtheorem{theorem}{Theorem}
\newtheorem{lemma}[theorem]{Lemma}

\newtheorem{definition}[theorem]{Definition}

\def\RR{{\mathbb R}}
\def\ZZ{{\mathbb Z}}

\def\PPP{{\mathcal P}}
\def\HHH{{\mathcal H}}
\def\VVV{{\mathcal V}}

\def\GGG{{\mathcal G}}
\def\QQQ{{\mathcal Q}}
\def\CCC{{\mathcal C}}

\def\Pr{{\mathrm Pr}}
\def \d{{\mathrm d}}

\def \e{{\mathrm e}}
\def \eps{\epsilon}
\def \notables{s}

\DeclareMathOperator*{\OO}{{\mathcal O}}

\DeclareMathOperator*{\diam}{\mathrm diam}
\DeclareMathOperator*{\poly}{\mathrm poly}

\colorlet{shade}{gray!5}
\newcommand{\algorithmus}[2]{
{\centering
\vspace{0.1cm}
\begin{tikzpicture}
\node [fill=shade,draw]
{\parbox{\textwidth-0.5cm}{
{\bf #1}\vspace{-0.5cm}\\
#2}
};
\end{tikzpicture}\\
\vspace{0.2cm}
}
}

\algrenewcommand{\algorithmiccomment}[1]{\hfill// #1}

\begin{document}
\title{ANN for time series under the Fr\'echet distance\thanks{We thank Karl Bringmann and André Nusser for useful discussions on the topic of this paper. Special thanks go to the anonymous reviewer who pointed out an error in an earlier version of the manuscript, and to Andrea Cremer for careful reading.}}

\author[1]{Anne Driemel} 
\author[1]{Ioannis Psarros}
\affil[1]{Hausdorff Center for Mathematics, University of Bonn, Germany\\
}
\affil[ ]{\textit{driemel@cs.uni-bonn.de}, \textit{ipsarros@uni-bonn.de}}
\maketitle             

\begin{abstract}
We study approximate-near-neighbor data structures for time series under the continuous Fr\'echet distance.
For an attainable approximation factor $c>1$ and a query radius $r$,  an approximate-near-neighbor data structure can be used to preprocess $n$ curves in $\mathbb{R}$ (aka time series), each of complexity $m$, to answer queries with a curve of complexity $k$ by either returning a curve that lies within Fr\'echet distance $cr$, or answering that there exists no curve in the input within distance $r$. In both cases, the answer is correct. 
Our first data structure achieves a $(5+\eps)$ approximation factor, uses space in 
$n\cdot  \mathcal{O}\left({\epsilon^{-1}}\right)^{k} + \mathcal{O}(nm)$ 
and has query time in $\mathcal{O}\left(k\right)$. Our second data structure achieves a $(2+\eps)$ approximation factor, uses space in 
$n\cdot  \mathcal{O}\left(\frac{m}{k\epsilon}\right)^{k} + \mathcal{O}(nm)$ 
and has query time in $\mathcal{O}\left(k\cdot 2^k\right)$.
Our third positive result is a probabilistic data structure based on locality-sensitive hashing, which achieves space in $\mathcal{O}(n\log n+nm)$ and query time in $\mathcal{O}(k\log n)$, and which answers queries with an approximation factor in $\mathcal{O}(k)$. All of our data structures make use of the concept of signatures, which were originally introduced for the problem of clustering time series under the Fr\'echet distance.
In addition, we show lower bounds for this problem. Consider any data structure which achieves an approximation factor less than $2$ and which supports curves of arclength up to $L$ and answers the query using only a constant number of probes. We show that under reasonable assumptions on the word size any such data structure needs space in $L^{\Omega(k)}$.

\end{abstract}

\section{Introduction}
For a long time, Indyk's result on approximate nearest neighbor algorithms for the discrete Fr\'echet distance of 2002~\cite{Indyk-approxnn-02} was the only result known for proximity searching under the Fr\'echet distance. However, recently there has been a raised interest in this area and several new results have been published~\cite{deBerg-2013fast, Gudmundsson-fqt-15, giscup2017, AD18, DPP-VC-19, Driemel-lshc-17,emiris2018, FFK20, aronov-nn-19, M20}. 
An intuitive definition of the Fr\'echet distance uses the metaphor of a person walking a dog. Imagine the dog walker being restricted to follow the path defined by the first curve while the dog is restricted to the second curve. In this analogy, the Fr\'echet distance is the shortest length of a dog leash that makes a dog walk feasible. 
Despite the many results in this area and despite the popularity of the Fr\'echet distance it is still an open problem how to build efficient data structures for it. 
Known results either suffer from a large approximation factor or high complexity bounds with dependency on the arclength of the curve, or only support a very restricted set of queries. Before we discuss previous work in more detail, we give a formal definition of the problem we study.

\begin{definition}[Fr\'{e}chet distance]
Given two curves $\pi,\tau:~ [0,1] \mapsto \RR $, their Fr\'{e}chet distance is:
\[
	\d_F(\pi,\tau)=\min_{\substack{f: [0,1] \mapsto [0,1] \\ g:[0,1] \mapsto [0,1]}} ~\max_{\alpha\in[0,1]} \|\pi(f(\alpha))- \tau (g(\alpha))\|_2,
\]
where $f$ and $g$ range over all continuous, non-decreasing functions with $f(0) = g(0) = 0$, and $f(1) = g(1) = 1$.
\end{definition}

 \begin{definition}[$c$-ANN problem]\label{Dgenann}
The input consists of $n$ curves $\Pi$ in $\RR^d$. Given a distance threshold $r>0,$ an approximation factor $c>1$, preprocess $\Pi$ into a data structure such that for any query $\tau$, the data structure reports as follows:
\begin{compactitem}
    \item if $\exists \pi\in \Pi$ s.t.~$\d_{F}(\pi,\tau)\leq r$, then it returns $\pi'\in \PPP$ s.t.~$\d_{F}(\pi,\tau)\leq cr$,
    \item if $\forall \pi \in \Pi$, $\d_{F}(\pi,\tau)\geq cr$ then it returns “no”,
    \item otherwise, it either returns a curve $\pi\in \Pi$ s.t.~$\d_{F}(\pi,\tau)\leq cr$, or  “no”. 
\end{compactitem}
\end{definition}

\subsection{Previous work}\label{sec:previouswork}
Most previous results on data structures for ANN search of curves, concern the \emph{discrete} Fr\'{e}chet distance. This is a simplification of the distance measure that only takes into account the vertices of the curves.
The first non-trivial ANN-data structure for the discrete Fr\'echet distance from 2002 by Indyk \cite{Indyk-approxnn-02} achieved approximation factor 
$\mathcal{O}((\log m + \log \log n)^{t-1})$, where $m$ is the maximum length of a sequence, and $t>1$ is a trade-off parameter. 
 More recently, in 2017, Driemel and Silvestri~\cite{Driemel-lshc-17} showed that locality-sensitive hashing can be applied and obtained a data structure of near-linear size which achieves approximation factor $\mathcal{O}(k)$, where $k$ is the length of the query sequence. They show how to improve the approximation factor to $\mathcal{O}(d^{3/2})$ at the expense of additional space usage (now exponential in $k$), and  a follow-up result by Emiris and Psarros~\cite{emiris2018} achieves a ($1+\eps$) approximation, at the expense of further increasing space usage. Recently,  
Filtser et al.~\cite{FFK20} showed how to build a $(1+\eps)$-approximate data structure using space in $ n \cdot \mathcal{O}(1/\eps)^{kd}$ and 
with query time in $\mathcal{O}(k d)$. 

These results are relevant in our setting, since the continuous Fr\'echet distance can naively be approximated using the discrete Fr\'echet distance.  
However, to the best of our knowledge, all known such methods introduce a dependency on the arclength of the curves (resp. the maximum length of an edge), either in the complexity bounds or in the approximation factor. It is not at all obvious how to avoid this when approximating the continuous with the discrete Fr\'echet distance. 

For the continuous Fr\'echet distance, a recent result by Mirzanezhad~\cite{M20} can be described as follows.
The main ingredient of this data structure is the discretization of the space of query curves with a grid, achieving an approximation factor of $1+\epsilon$. Alas, the space required for each input curve is high, namely roughly $D^{dk}$, where $D$ is the diameter of the set of vertices of the input, $d$ is the dimension of the input space and $k$ is the complexity of the query.  

Interestingly, there are some data structures for the related problem of {range searching}, which are especially tailored to the case of the continuous Fr\'echet distance and which do not have a dependency on the arclength. 
The subset of input curves, that lie within the search radius of the query curve is called the {range} of the query.  
A range query should return all input curves inside the range, or a statistic thereof. 
Driemel and Afshani \cite{AD18} consider the exact range searching problem for polygonal curves under the Fr\'echet distance. For $n$ curves of complexity $m$ in $\RR^2$, their data structure uses space in $\OO\left(n (\log \log n)^{\OO(m^2)}\right)$ and the query time costs 
$\OO\left(\sqrt{n} \log^{\OO(m^2)} n\right)$, assuming that the complexity of the query curves is at most $\log^{\OO(1)} n$. They also show lower bounds in the pointer model of computation that match the number of log factors used in the upper bounds asymptotically. The new lower bounds that we show in this paper also hold for the case of range searching (more specificially, range emptiness queries), but we assume a different computational model, namely the cell-probe model. While the lower bound of Afshani and Driemel only holds in the case of exact range reporting and uses curves in the plane, our new lower bound also holds in the case of approximation and is meaningful from $d \geq 1$.

\subsection{Known techniques}

Our techniques are based on a number of different techniques that were previously used only for the discrete Fr\'echet distance. 
In this section we give an overview of these techniques and highlight the main challenges that distinguish the discrete Fr\'echet distance from the continuous Fr\'echet distance.

The locality-sensitive hashing scheme proposed by Driemel and Silvestri \cite{Driemel-lshc-17} achieves linear space and query time in $\OO(k)$, with an approximation factor of $\OO(k)$ for the \emph{discrete} Fr\'echet distance. The data structure is based on snapping vertices to a randomly shifted grid and then removing consecutive duplicates in the sequence of grid points produced by snapping. Any two near curves produce the same sequence of grid points with constant probability while any two curves, which are sufficiently far away from each other, produce two non-equal sequences of grid points with certainty. The main argument used in the analysis of this scheme involves the optimal discrete matching of the vertices of the two curves. This analysis is not directly applicable to the continuous Fr\'echet distance as the optimal matching is not always realized at the vertices of the curves.  

There are several ANN data structures with fast query time and small approximation factor which store a set of representative  query candidates together with precomputed answers for these queries so that a query can be answered approximately with a lookup table. One example of this is the $(1+\eps)$-ANN data structure for the $\ell_{p}$ norms \cite{HIM12}, which employs a grid and stores all those grid points which are near to some data point, and a pointer to the data point that they represent. The side-length of the grid controls the approximation factor and using hashing for storing precomputed solutions leads to an efficient query time. A similar approach was used by Filtser et~al.~\cite{FFK20} for the $(1+\eps)$-ANN problem under the \emph{discrete} Fr\'echet distance. The algorithm discretizes the 
query space with a canonical grid and stores representative point sequences on this grid. 

There are several challenges when trying to apply the same approach to the ANN problem under the continuous Fr\'echet distance. 
Computing good representatives in this case is more intricate: two curves may be near but some of their vertices may be far from any other vertex on the other curve. Hence, picking representative curves which are defined by vertices in the proximity of the vertices of the data curve is not sufficient. In case the input consists of curves with bounded arclength only, 
one can enumerate all curves which are defined by grid points and lie within a given Fr\'echet distance. However, this results in a large dependency on the arclength.

The question, whether efficient ANN data structures for the continuous Fr\'echet distance which do not have a dependency on the arclength of the input curves are possible, is an intriguing question, which we attempt to answer in this paper.

\subsection{Preliminaries}\label{sec:prelims}

For any $x\in \RR$, $|x|$ denotes the absolute value of $x$. For any positive integer $n$, $[n]$ denotes the set $\{1,\ldots,n\}$. 
Throughout this paper, a \emph{curve} is a continuous function $[0,1]\mapsto \RR$ and we may refer to such a curve as a \emph{time series}. We can define a curve $\pi$ as $\pi:=\langle x_1,\ldots,x_m\rangle$, which means that $\pi$ is obtained by linearly interpolating $x_1,\ldots,x_m$. The \emph{vertices} of $\pi: [0,1]\mapsto \RR$ are those points  which are local extrema in $\pi$. 
For any curve $\pi$, $\VVV(\pi)$ denotes the sequence of vertices of $\pi$. 
The number of vertices $|\VVV(\pi)|$ is called the \emph{complexity} of $\pi$ and it is also denoted by $|\pi|$.   For any two points $x$, $y$, $\overline{xy}$ denotes the directed line segment connecting $x$ with $y$ in the direction from $x$ to $y$. The segment defined by two consecutive vertices is called an \emph{edge}. 
For any two $0\leq p_a<p_b\leq1$ and any curve $\pi$, we denote by $\pi[p_a,p_b]$ the subcurve of $\pi$ starting at $\pi(p_a)$ and ending at $\pi(p_b)$.  
For any two curves $\pi_1$, $\pi_2$, with vertices $x_1,\ldots,x_k$ and $x_k,\ldots,x_m$ respectively, 
$\pi_1 \oplus \pi_2 $ denotes the curve $\langle x_1,\ldots,x_k, \ldots x_m \rangle  $, that is the concatenation of $\pi_1$ and $\pi_2$. 
We define the \emph{arclength} $\lambda(\pi)$ of a curve $\pi$ as the total sum of lengths of the edges of $\pi$.
We refer to a pair of continuous, non-decreasing  functions $f:[0,1]\mapsto [0,1]$, $g: [0,1]\mapsto [0,1]$ such that $f(0)=g(0)$, $f(1)=g(1)$, as a \emph{matching}. If a matching $\phi=(f,g)$ of two curves $\pi$, $\tau$  satisfies $\max_{\alpha\in[0,1]}\|\pi(f(\alpha))-\tau(g(\alpha)) \|\leq \delta$, then we say that $\phi$ is a $\delta$-matching of $\pi$ and $\tau$. Given two curves $\pi:[0,1]\rightarrow \RR$, $\tau:[0,1]\rightarrow \RR$. The $\delta$-free space is the subset of the parametric space $[0,1]^2$ defined as 
$\{ (x,y) \in [0,1]^2 \mid | \pi(x) - \tau(y)| \leq \delta \}$.

Our data structures make use of a \emph{dictionary} data structure. A dictionary stores a set of (key, value) pairs and when presented with a key, returns the associated value. Assume we have to store $n$ (key,value) pairs, where the keys come from a universe $U^k$. Perfect hashing provides us with a dictionary using $O(n)$ space and $O(k)$ query time which can be constructed in  $O(n)$ expected time~\cite{FKS84}. During look-up, we compute the hash function in $\OO(k)$ time, we access the corresponding bucket in the hashtable in $\OO(1)$ time and check if the key stored there is equal to the query in $\OO(k)$ time.

All of our data structures operate in the \emph{real-RAM model}, enhanced with floor function operations in constant time. See also Appendix~\ref{sectionappendix:compmodels} for a more detailed discussion on the computational models. Our lower bounds are for the cell-probe model.
The cell-probe model of computation counts the number of memory accesses (cell probes) to the data structure which are performed by a query. Given a universe of data and a universe of queries, a cell-probe data structure with performance parameters $s$, $t$, $w$, is a structure which consists of $s$ memory cells, each able to store $w$ bits, and any query can be answered by accessing $t$ memory cells.
Our lower bound concerns approximate distance oracles. A Fr\'echet distance oracle is a data structure which, given one input curve $\pi$, a distance threshold $r$, and an approximation factor $c>0$, it reports for any query curve $\tau$ as follows: \begin{enumerate*}[label=(\roman*)]
\item if $\d_F(\pi,\tau)\leq r$ then the answer is “yes”,
\item if $\d_F(\pi,\tau)>c r$ then the answer is “no”,
\item otherwise the answer can be either “yes” or “no”.
\end{enumerate*}

The standard algorithm by Alt and Godau \cite{altgodau-95} for computing the Fr\'echet distance between two curves $\pi,\tau$, finds a matching in the parametric space of the two curves, where a matching is realized by a monotone path which starts at $(0,0)$ and ends at $(1,1)$. If such a path is entirely contained in the $\delta$-free space, then $\d_F(\pi,\tau)\leq\delta$. 
The Fr\'echet distance is known to satisfy the triangle inequality. We use the following two observations repeatedly in the paper.
\begin{enumerate}
    \item[(i)] For any curves $\tau_1,\tau_2$, $\pi_1,\pi_2$, 
which satisfy the property that the last vertex of $\tau_1$ is the first vertex of $\tau_2$ and the last vertex of $\pi_1$ is the first vertex of $\pi_2$, 
it holds that 
$\d_F(\tau_1\oplus \tau_2, \pi_1 \oplus \pi_2) \leq \max\{\d_F(\tau_1,\tau_2),\d_F(\pi_1,\pi_2) \}$. 
\item[(ii)] For any two edges $\overline{a_1a_2}$, $\overline{b_1b_2}$, it holds that  $\d_F(\overline{a_1a_2},\overline{b_1b_2})=\max\{|a_1-b_1|,|a_2-b_2|\}$. 
\end{enumerate}
 These two facts imply that if $\pi_1=\langle x_1,\ldots,x_m \rangle$ and $\pi_2=\langle y_1,\ldots,y_m \rangle$ such that for each $i=1,\ldots,m$, $|x_i-y_i|\leq \eps$ then $\d_F(\pi_1,\pi_2)\leq \eps$. This is a key property that we exploit when we snap vertices of a curve to a grid, since it allows us to bound the distance between the original curve, and the curve defined by the sequence of snapped  vertices. 

We end this section with the standard definition of the discrete Fr\'echet distance. For any positive integer $m$, $\left(\RR^{d}\right)^m$ denotes the space of sequences of $m$ real  vectors of dimension $d$. 

 \begin{definition}[Traversal]
 Given $P=p_1, \ldots, p_{m}\in \left(\RR^d\right)^m$ and $Q=q_1, \ldots, q_{k}\in \left(\RR^d\right)^k$, a traversal $T=(i_1,j_1),\ldots,(i_t,j_t)$ of $P$ and $Q$ is a sequence of pairs of indices referring to a pairing of points from the two sequences such that:
\begin{compactenum}[(i)]
 \item $i_1,j_1=1$, $i_t=m$, $j_t=k$. 
 \item $\forall (i_u, j_u)\in T:$ $i_{u+1}-i_u \in \{0,1\}$ and $j_{u+1}-j_u \in \{0,1\}$.
 \item $\forall (i_u, j_u)\in T:$ $(i_{u+1}-i_u)+(j_{u+1}-j_u)\geq1$.
\end{compactenum} 
 For any traversal $T$, we define $\d_T(P,Q):= \max_{(i,j)\in T} \|p_i -q_j\|_2$.
 \end{definition} 

\begin{definition}[Discrete Fr\'{e}chet distance]\label{Ddist}
 Given $P=p_1,\ldots,p_m \in \left(\RR^d\right)^m$ and $Q=q_1,\ldots,q_k\in \left(\RR^d\right)^k$, we define the discrete Fr\'{e}chet distance between $P$ and $Q$ as follows:
 \[
 \d_{dF}(P,Q)= \min_{T\in\mathcal{T}}  \max_{(i_u,j_u)\in T} \|p_{i_u}-q_{j_u}\|_2  ,
 \]
 where $\mathcal{T}$ denotes the set of all possible traversals for $P,Q$. Thus, $\d_{dF}(P,Q)=\min_{T\in \mathcal{T}} \d_T(P,Q)$. 
 \end{definition}

\subsection{Our contributions}
We study the $c$-ANN problem for time series under the continuous Fr\'echet distance. 
Our first result is data structure that achieves approximation factor $5+\eps$ for any $\eps>0$. The data structure is described in Section~\ref{section:firstds} and leads to the following theorem.

\begin{restatable}{theorem}{thmfirstANNFD}\label{theorem:firstconstantapprx}
Let $\eps\in(0,1]$. There is a data structure for the $(5+\eps)$-ANN problem, 
which stores $n$ time series of complexity $m$ and supports query time series of complexity $k$, which uses space in $n\cdot \OO\left(\frac{1}{\eps}\right)^k+\OO(nm)$, needs $\OO\left(n m \right)\cdot \OO\left(\frac{1}{\eps}\right)^k$ expected  preprocessing time and answers a query in $\OO(k)$ time. 
\end{restatable}
To achieve this result, we generate a discrete approximation  of the set of all possible non-empty queries. To this end, we employ the concept of signatures, previously introduced in~\cite{DKS16}. The signature of a time series provides us with a selection of the local extrema of the function graph, which we use to approximate the set of queries. 

We extend these ideas to improve the approximation factor to $(2+\eps)$, albeit with an increase in space and query time. In particular, we generate all curves with vertices that lie in the vicinity of the vertices of the input curves. We combine this with a careful analysis of the involved matchings and a more elaborate query algorithm. 
The resulting data structure can be found in Section \ref{section:secondds} and leads to the following theorem.

\begin{restatable}{theorem}{thmANNFD}\label{theorem:secondconstantapprx}
Let $\eps\in(0,1]$. There is a data structure for the $(2+\eps)$-ANN problem, 
which stores $n$ time series of complexity $m$ and supports query time series of complexity $k$, which uses space in $n\cdot \OO\left(\frac{m}{k \eps}\right)^{k}$, needs $\OO(nm)\cdot \OO\left(\frac{m}{k \eps}\right)^{k}$ expected preprocessing time and answers a query in $\OO(k\cdot 2^{k})$ time. 
\end{restatable}

Our third result is a data structure that uses space in $\OO(n\log n+nm)$  and has query time in $\OO(k\log n)$. This improvement in the space complexity comes with a sacrifice in the approximation factor achieved by the data structure, which is now in $\OO(k)$.

\begin{restatable}{theorem}{thmANNHA}
\label{theorem:largeapproximation}
There is a data structure for the $(24k+1)$-ANN problem, 
which stores $n$ time series of complexity $m$ and supports queries with time series of complexity $k$,  uses  space in $\OO(n\log n + nm)$, needs $\OO(nm \log n)$ expected preprocessing time and answers a query in $\OO(k \log n)$ time. For a fixed query, the preprocessing succeeds with probability at least $1-1/\poly(n)$. 
\end{restatable}

To achieve this result, we combine the notion of signatures with the ideas of the locality-sensitive scheme that was previously used \cite{Driemel-lshc-17} for the discrete Fr\'echet distance. In the discrete case, it is sufficient to snap the vertices of the curves to a grid of well-chosen resolution and to remove repetitions of grid points along the curve to obtain a hash index with good probability. In the continuous case, we first compute a signature, which filters the salient points of the curve, and only then apply the grid snapping to this signature to obtain the hash index. The resulting data structure is surprisingly simple. The description of the data structure can be found in Section~\ref{sec:lsh}.

Finally, we give a lower bound in the cell-probe model of computation, which seems to indicate that for data structures that achieve approximation factor better than $2$ and that use a constant number of probes per query, a dependency on the arc-length of the curve is necessary.

\begin{restatable}{theorem}{thmDOLB}\label{thm:lowerbound}
Consider any Fr\'echet distance oracle with approximation factor ${2}-\gamma$, for any $\gamma\in(0,1]$,   distance threshold $r=1$,
in the cell-probe model, which supports time series  as follows: it stores any  polygonal curve in $\RR$ of arclength  at most $L$, for $L\geq 6$, it  supports  queries of arclength up to $L$ and complexity $k$, where $k\leq L/6$, 
and it achieves performance parameters $t$, $w$, $s$. There exist 
\begin{align*}
w_0=\Omega \left(\frac{{L}^{1-\eps}}{t} \right),& &     s_0&= 2^{\Omega \left(\frac{ k \log(L/k)}{t} \right)}
\end{align*}
such that if $w<w_0$
then 
$s\geq s_0$, for any constant $\eps>0$.
\end{restatable}

To achieve this result we observe that a technique first introduced by Miltersen \cite{M94} can be applied here. Miltersen shows that lower bounds for communication problems can be translated into lower bounds for cell-probe data structures. In particular, we use a reduction from the lopsided disjointness problem (see Section~\ref{sec:lowerbound_reduction}). 

In addition, we extend these lower bound results to the case of the discrete Fr\'echet distance. Here, our reduction is more intricate. We adapt a reduction by Bringmann and Mulzer~\cite{BM16}, which was used for showing lower bounds for computing the Fr\'echet distance.
Our results show that an exponential dependence on $k$  for the space is necessary when the number of probes is constant (such as in \cite{FFK20}).

\subsection{Signatures}
A crucial ingredient to our algorithms is the notion of \emph{signatures} which was first introduced in \cite{DKS16}. We define signatures as follows.

\begin{definition}[$\delta$-signatures]
\label{def:sig}
A curve $\sigma :~ [0, 1] \mapsto \RR$ is a $\delta$-signature of $\tau:~ [0, 1] \mapsto \RR$ if it is a curve   
defined by
a series of values $0 = t_1 < \dots < t_{\ell} = 1$ as the linear
interpolation of $
\tau(t_i)$ in the order of the index $i$, and
satisfies the following properties.
For $1 \leq i \leq \ell- 1$ the following conditions hold:
\begin{enumerate}
    \item[i)](non-degeneracy) if $i \in [2, \ell - 1]$ then $\tau (t_i) \notin \overline{\tau(t_{i-1}), \tau(t_{i+1})}$,
    \item[ii)]  (direction-preserving) if $\tau(t_i) < \tau(t_{i+1})$ for $t < t' \in [t_i, t_{i+1}]$:
$\tau (t) - \tau (t')\leq 2\delta$ and
if $\tau (t_i) > \tau (t_{i+1})$ for $t < t'\in [t_i, t_{i+1}]$:
$\tau (t') - \tau (t) \leq 2\delta$,
\item[iii)] (minimum edge length)
if $i \in [2, \ell- 2]$ then $|\tau (t_{i+1}) - \tau (t_i)| > 2\delta$, and if $i \in \{1, \ell - 1\}$ then $|\tau (t_{i+1}) - \tau (t_i)| > \delta$,
\item[iv)](range) for 
$t \in [t_i, t_{i+1}]:$
if $i \in [2, \ell- 2]$ then $\tau (t) \in \overline{\tau (t_i) \tau (t_{i+1})}$, and
if $i = 1$ and $\ell > 2$ then $\tau (t) \in \overline{\tau (t_i)  \tau(t_{i+1})} \cup
\overline{(\tau (t_i) - \delta) (\tau (t_i) + \delta)}$, and
if $i = \ell - 1$ and $\ell > 2$ then $\tau (t)\in
\overline{\tau (t_{i-1}) \tau (t_i)} \cup \overline{(\tau (t_i) - \delta) (\tau (t_i) + \delta})$, and 
if $i = 1$ and $\ell = 2$ then $\tau (t) \in  \overline{\tau (t_1) \tau (t_2)} \cup
\overline{(\tau (t_1) -\delta) (\tau (t_1) + \delta)} \cup \overline{(\tau (t_2) -\delta )(\tau (t_2) + \delta)}$.
\end{enumerate}
\end{definition}
For any $\delta>0$ and any curve $\pi:~[0,1]\mapsto \RR$ of complexity $m$, a $\delta$-signature of $\pi$ can be computed in $\OO(m)$ time \cite{DKS16}. We now state some basic results about signatures. 

\begin{lemma}[Lemma 3.1 \cite{DKS16}]
\label{lemma:signatures1}
It holds for any $\delta$-signature $\sigma$ of $\tau$:
$d_F (\sigma, \tau ) \leq \delta$.
\end{lemma}

\begin{lemma}[Lemma 3.2 \cite{DKS16}]\label{lemma:signatures2}
Let $\sigma$  with vertices $v_1, \ldots , v_{\ell} $, be a $\delta$-signature of
$\pi$ with vertices $ u_1, \ldots , u_m $. Let $r_i = [v_i-\delta, v_i+\delta]$, for $1 \leq i \leq \ell$,
be ranges centered at the vertices of $\sigma$ ordered along
$\sigma$. It holds for any time series $\tau$ if $d_F (\pi,\tau) \leq \delta$, then $\tau$ has
a vertex in each range $r_i$, and such that these vertices
appear on $\tau$ in the order of $i$.
\end{lemma}

\section{A constant-factor approximation for time  series}
\label{section:firstds}
In this section, we describe the data structure for  Theorem~\ref{theorem:firstconstantapprx}. The data structure achieves approximation factor  $(5+\eps)$.

\paragraph{The data structure} 
The input consists of a set $\Pi$ of $n$ curves in $\RR$, and the approximation error $\eps>0$. 
To simplify our exposition, we assume that the distance threshold $r$ is equal to $1$ (otherwise, we scale the input uniformly). 
To solve the problem for a different value of $r$, the input set can be uniformly scaled. 
Let $\GGG_{w}:=\{i\cdot w \mid i\in \ZZ\}$ be the regular grid with side-length $w:= \eps/2$. Let $\HHH$ be a dictionary, which is initially  empty. 
For each input curve $\pi \in \Pi$, we compute its  $1$-signature $\sigma_{\pi}$, with vertices $\VVV(\sigma_{\pi})=v_1,\ldots,v_{\ell}$, and for each $v_i\in \VVV(\sigma_{\pi})$ we define the range  $r_i:=[v_i- 2-w,v_i+2+w]$. We enumerate all curves with at most $k$ vertices, chosen from the sets $r_1\cap G_w, r_2\cap G_w, \ldots$, and satisfying the order of $i$, and we store them in a set $\CCC'$. Next, we compute the set $\CCC(\pi):=\{\sigma \in \CCC' \mid \d_F(\sigma,\pi)\leq 3\}$. We store $\CCC(\pi)$ in $\HHH$ as follows: for each $\sigma \in \CCC(\pi)$, we use as key   the sequence of its vertices $\VVV(\sigma)$: if $\VVV(\sigma)$ is not already stored in $\HHH$, then we insert the pair ($\VVV(\sigma)$,$\pi$)
into $\HHH$. 
The total space required is $\OO\left(n\cdot \max_{\pi \in \Pi} |\CCC(\pi)|\right)$. 

Our intuition is the following. We would like  the set $\CCC(\pi)$ to contain all those curves that correspond to $2$-signatures of query curves that have $\pi$ as an  approximate near neighbor in the set $\Pi$. So when presented with a query we can simply compute its $2$-signature and do a lookup in  $\HHH$. However, the set of all possible $2$-signatures with non-empty query is infinite. Therefore, we snap the vertices to a grid to obtain a discrete set of bounded size.
\paragraph{The query algorithm} 
When presented with a query curve $\tau$, we first compute a $2$-signature $\sigma_{\tau}$, and then we compute a key by snapping the vertices to the same grid $\GGG_{w}$. Snapping to $\GGG_w$ is implemented as follows: if $\VVV(\sigma_{\tau})= v_1,\ldots,v_{\ell} $ then $\sigma_{\tau}':=\langle g_w(v_1),\ldots,g_w(v_{\ell}) \rangle$, where for any $x\in \RR$, $g_w(x)$ is the nearest point of $x$ in $\GGG_w$. 
We perform a lookup in  $\HHH$ with the key $\VVV(\sigma_{\tau}')$ and return the result: if $\VVV(\sigma_{\tau}')$ is stored in $\HHH$ then we return the associated curve, otherwise we return “no”.

\begin{lemma}\label{lem:5correctness}
Let $\tau$ be a query curve of complexity $k$. If the query algorithm returns an input curve $\pi' \in \Pi$, then $\d_F(\pi',\tau)\leq 5+\eps$. If the query algorithm returns ``no'', then there is no $\pi \in \Pi$ such that $\d_F(\pi,\tau) \leq 1$. 
\end{lemma}
\begin{proof}
Let $\pi$ be any input curve in $\Pi$ and let $\sigma_{\pi}$ be the $1$-signature of $\pi$. Let $\tau$ be a query curve, let $\sigma_{\tau}$ be its $2$-signature and let $\sigma_{\tau}'$ be as defined in the query algorithm. 
First suppose that $\d_F(\pi,\tau)\leq 1$. 
By the triangle inequality and Lemma \ref{lemma:signatures1}, $\d_F({\pi},\sigma_\tau)\leq 3+w$. Let $ u_1,\ldots,u_{\ell'}$ be the vertices of  $\sigma_{\tau}$ and define for each $i\in [\ell']$, $r_i':=[u_i - 2,u_i + 2]$.  
By Lemma \ref{lemma:signatures2}, $\sigma_{\pi}$ has a vertex  in each range $r_i'$ and these vertices appear on $\sigma_{\pi}$ in the order of $i$. This guarantees that the vertices of $\sigma_{\tau}'$ lie in the ranges $r_1,\ldots,r_{\ell}$ and it will be considered during preprocessing. Hence, $\sigma_{\tau}'$ will be generated when preprocessing $\pi$.
This implies that $\VVV(\sigma_{\tau}')$ is stored in $\HHH$. 
It is possible that $\sigma_{\tau}'$ was also generated and stored for a different input curve, say $\pi' \neq \pi$ with $\d_F(\pi',\sigma_{\tau}') \leq 3$. We claim that $\d_F(\pi',\tau) \leq 5+2w$. Indeed, we have by the triangle inequality
\[ \d_F(\pi',\tau) \leq \d_F(\pi',\sigma_{\tau}') + \d_F(\sigma_{\tau}',\sigma_{\tau}) + \d_F(\sigma_{\tau},\tau) \leq 5+2w. \]
This proves that any curve returned by the query algorithm has Fr\'echet distance at most $5+2w=5+\eps$ to the query curve, and if the query algorithm returns ``no'', then there is no input curve within Fr\'echet distance $1$ to the query curve.

  \end{proof}

\thmfirstANNFD*
\begin{proof}
The data structure is described above. By Lemma~\ref{lem:5correctness} the data structure returns a correct result. It remains to analyze the complexity. 
Our data structure solves the $(5+\eps)$-ANN problem with distance threshold $r=1$. 
The space required for each input curve is proportional to the number of candidate signatures computed in the preprocessing phase. 
Indeed, we will show now that $|\CCC'|\leq \OO\left( \frac{1}{\eps}\right)^k$. 
Notice that if there exists a curve with $k$ vertices which is within distance $1$ from $\pi$ then $\ell\leq k$, by Lemma \ref{lemma:signatures2}. Recall that the curves in $|\CCC'|$ have vertices in the ranges $r_i\cap \GGG_{w}$ and the vertices  respect the order of $i$. In particular, \texttt{generate\_sequences} adds at most one curve to $\CCC'$ for each possible sequence of vertices in   $r_i\cap \GGG_{w}$, $i=1,\ldots,\ell$, that satisfy the order of $i$. 
If we fix the choices of  $t_1,\ldots,t_{\ell}$, where each $t_i$ denotes the number of vertices in $r_i\cap \GGG_{w}$ to be  used in the creation of those curves, we can produce at most $\prod_{i=1}^{\ell} |r_i\cap \GGG_{w}|^{t_i}$ distinct sequences of vertices of length $\sum_{i=1}^{\ell} t_i$ and hence at most $\prod_{i=1}^{\ell} |r_i\cap \GGG_{w}|^{t_i}$ curves of length at most $\sum_{i=1}^{\ell} t_i$. 
Hence, 
\begin{align*}
|\CCC'| &\leq
\sum_{\substack{t_1+\ldots +t_{\ell}=k \\ \forall i:~ t_i\geq 0 \\t_1\geq 1,t_{\ell}\geq 1 }} \prod_{i=1}^{\ell} \left(\frac{4}{\eps}+2\right)^{t_i } \\ & \leq 
\sum_{\substack{t_1+\ldots +t_{\ell}=k \\ \forall i:~ t_i\geq 0}} \left(\frac{4}{\eps}+2\right)^{k} \\ &\leq {{k+\ell-1}\choose{k}} \cdot
\left(\frac{4}{\eps}+2\right)^{k}\\
& \leq (2\e)^{k} \cdot \left(\frac{4}{\eps}+2\right)^{k}\\
&=\OO\left(\frac{1}{\eps}\right)^{k}.    
\end{align*}

The time to compute a signature for a curve of complexity $m$ is $\OO(m)$, because we can use the algorithm of \cite{DKS16}. For filtering out the candidates with high Fr\'echet distance we apply the decision algorithm by Alt and Godau~\cite{altgodau-95} for each candidate in $O(mk)$ time.  
Since we rely on perfect hashing for building $\HHH$, the expected preprocessing time is in 
$\OO(nm) \cdot \OO(1/\eps)^k$, and the space is in  $n\cdot \OO(1/\eps)^k$ because we can store pointers to curves in $\HHH$, plus $\OO(nm)$ for storing the input curves. 
Each query costs $\OO(k)$ time, since we employ perfect hashing for $\HHH$, and snapping a curve costs $\OO(k)$ time assuming that a floor function operation needs $\OO(1)$ time. 

  \end{proof}

Deciding whether a query curve $\tau$ is near to a given curve $\pi$ by only having a $2$-signature of $\tau$ is subject to a $\pm2$ error. 

One can find concrete worst-case examples where this  approximation factor is attained.

\subsection{Pseudocode of the basic result}\label{app:pseudo}
\small

\algorithmus{\texttt{preprocess}(set of time series $\Pi$, $\eps>0$)}{

\begin{algorithmic}[1]
\Statex \Comment{$k$ is assumed to be a variable with global scope}
     \State Initialize empty dictionary $\HHH$ 
     \State $w\gets \eps/2$ 
     \For {{\bf each} $\pi \in \Pi$}
     \State $\CCC(\pi)\gets$ \texttt{generate\_candidates}($\pi$, $w$)
     \If{ $\CCC(\pi)\neq \emptyset$}
     \For{{\bf each} $\sigma_{\tau}\in \CCC(\pi)$}
     \If{$\VVV(\sigma_{\tau})$ not in $\HHH$}
     \State insert key $\VVV(\sigma_{\tau})$ in $\HHH$, associated with a pointer to $\pi$  
     \EndIf
     \EndFor
     \EndIf
     \EndFor
\end{algorithmic}
}

    \algorithmus{\texttt{generate\_candidates}(time series $\pi$, $w>0$)}{
\begin{algorithmic}[1]
\State $\sigma_{\pi}\gets $ $1$-signature of $\pi$, with  $\VVV(\sigma_{\pi})= v_1,\ldots,v_{\ell}  $
 \If{ $\ell>k$} 
 \State  {\bf return} $\emptyset$
 \EndIf
\For{ {\bf each} $i=1,\ldots,\ell$}
\State $ r_i\gets  [v_i-2-w,v_i+2+w] $
\EndFor
\State $\CCC' \gets \emptyset$ 
\For{{\bf each} $j=1,\ldots, \ell $}
\For{{\bf each} $p \in r_j\cap \GGG_w$}
\State \texttt{generate\_sequences}($\langle p\rangle$, $j$, $w$, $\CCC'$) 
\EndFor
\EndFor
\State $\CCC (\pi) \gets \emptyset$
\For{ {\bf each} $\sigma_{\tau}\in \CCC'$}
\If{$\d_F({\pi},\sigma_{\tau})\leq 3+w$} 
\State $\CCC(\pi) \gets \CCC(\pi) \cup \{ \sigma_{\tau}\}$
\EndIf
\EndFor
\State {\bf  return} $\CCC (\pi) $
\end{algorithmic}
}

 \algorithmus{\texttt{generate\_sequences}(time series $\sigma$, integer $i$, $w>0$, returned set $\CCC'$)}{
\begin{algorithmic}[1]
\Statex \Comment{Stores in $\CCC'$ all possible time series which begin with $\sigma$, have at most $k$ vertices that belong to  $r_j\cap \GGG_{w}$, for $j=i,\ldots, \ell$, and appear in them in the order of $j$.}
\State $v_1,\ldots,v_t \gets \VVV(\sigma)$
\If {$|\VVV(\sigma)|\leq k$}
\State $\CCC'\gets \CCC' \cup \{ \sigma\}$
\EndIf
\If{$|\VVV(\sigma)|<k$}
\For{{\bf each} $j=i,\ldots, \ell $}
\For{{\bf each} $p \in r_j\cap \GGG_w$}
\State $\sigma' \gets \langle v_1,\ldots,v_t,p\rangle $ 
\State \texttt{generate\_sequences}($\sigma'$, $j$, $w$, $\CCC'$)
\EndFor
\EndFor
\EndIf
\end{algorithmic}
}

\algorithmus{\texttt{query}(time series $\tau$)}{

\begin{algorithmic}[1]
    \Statex \Comment{$w= \eps/2$ is fixed during preprocessing}
     \State $\sigma_{\tau} \gets $ compute a $2$-signature of $\tau$.
  \State $\sigma_{\tau}' \gets$ snap $\sigma_\tau$ to $\GGG_{w}$
     \If{ $\exists \pi \in \Pi$,  $ \sigma_{\tau}'\in \CCC(\pi)$ }  \Comment{lookup in  $\HHH$}
     \State{\bf report} $\pi$ \Comment{arbitrary $\pi$ s.t.\ $ \sigma_{\tau}'\in \CCC(\pi)$ }
  \Else 
    \State {\bf report} “no”
  \EndIf
\end{algorithmic}
}

\normalsize

\section{Improving the approximation factor to $(2+\eps)$}
\label{section:secondds}

In this section, we describe the data structure for  Theorem~\ref{theorem:secondconstantapprx}.  
We build upon the ideas developed in Section \ref{section:firstds}. 
The key to circumventing the larger approximation factor resulting from the use of the triangle inequality seems to be a careful construction of matchings. 
For this we define the notion of a $\delta$-tight matching for two curves.

\begin{figure}[t]
    \centering
    \includegraphics[width=\textwidth, page=3]{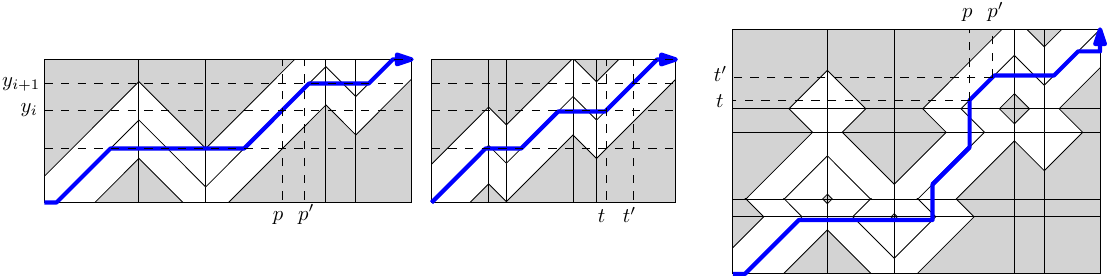}
    \caption{Example of the path constructed in the proof of Lemma~\ref{lemma:lefttightmatching}. The left figure shows a tight matching from $X$ to $\pi$. The middle figure shows a tight matching from $X$ to $\tau$. Diagonal edges of the $0$-free space of these can be transferred to the diagram on the right, which is the free space diagram of $\pi$ and $\tau$. The final path results from connecting these diagonal segments using horizontal and vertical line segments.  }
    \label{fig:tight}
\end{figure}

\subsection{Tight matchings}\label{section:tightmatchings}

Intuitively, a $\delta$-tight matching is a matching which attains a distance of at most $\delta$ and  matches as many pairs of points as possible at distance zero.

\begin{definition}[$\delta$-tight matching]
\label{definition:tightmatching}
Given two curves $\pi$ and $\tau$, consider a monotone path $\lambda$ through the parametric space of $\pi$ and $\tau$ consisting of two types of segments: \begin{compactenum}[(i)]
            \item a segment contained in the $0$-free space (corresponding to identical subcurves of $\pi$ and  $\tau$),
            \item a horizontal line segment contained in the $\delta$-free space (corresponding to a point on $\pi$ and a subcurve on $\tau$). 
               \end{compactenum}
        If $\lambda$ exists, we say $\lambda$ is a \emph{tight matching} of width $\delta$ from $\pi$ to $\tau$. 
\end{definition}

Full proofs of the lemmas of this section can be found in Section~\ref{app:proofs}.

We use the following theorem by Driemel, Krivosija and Sohler. 

\begin{theorem}[Theorem 3.1~\cite{DKS16}]
\label{theorem:deletionofvertices}
Let $\sigma_{\tau}$ be a $\delta$-signature of
$\tau$ with vertices  $v_1, \ldots , v_{\ell}$. Let $r_j = [v_j - \delta, v_j + \delta]$ be ranges
centered at the vertices of $\sigma_{\tau}$ ordered along $\sigma_{\tau}$, where
$r_1 = [v_1 - 4\delta, v_1 + 4\delta]$ and $r_{\ell} = [v_{\ell} - 4\delta, v_{\ell} + 4\delta]$. Let
$\pi$ be a curve with $\d_F (\tau, \pi) \leq \delta$ and let $\pi'$ be a curve obtained by removing some vertex $u_i = \pi(p_i)$ from $\pi$ with 
$u_i \notin \bigcup_{1\leq j \leq \ell}r_j$. It holds that $\d_F (\tau, \pi') \leq \delta$.
\end{theorem}

\begin{restatable}{lemma}{lemmalefttightmatchingfirst}
\label{lemma:lefttightmatching1}
Let $X=\overline{a b} \subset \RR$ be a line segment and let $\tau:[0,1]\rightarrow \RR$ be a curve with $[a,b] \subseteq [\tau(0), \tau(1)]$. If  $\d_F(X, \tau)\leq \delta$ then there exists a $\delta$-tight matching from $X$ to $\tau$.
\end{restatable}

\begin{proof}[Proof Sketch]
We first construct a connected path in the $\delta$-free space of the two curves that only consists of sections of the $0$-free space and horizontal line segments, but is not necessarily monotone.
We do this by parametrizing the set that constitutes the $0$-free space and connecting it by horizontal line segments. We obtain an $x$-monotone connected curve from $(0,0)$ to $(1,1)$ which lies inside the $\delta$-free space. We then show that this path can be iteratively ``repaired'' by replacing non-monotone sections of the path with  horizontal segments, while maintaining the property that the path is contained inside the $\delta$-free space. After a finite number of iterations of this procedure we obtain a $\delta$-tight matching from $X$ to $\tau$. Figure~\ref{fig:lefttight} illustrates the process.
  \end{proof}
 
In the next lemma we combine tight matchings from a line segment to show an upper bound on the Fr\'echet distance. Using this lemma, we can show upper bounds on the distance that are stronger than bounds obtained by triangle inequality. Figure~\ref{fig:tight} illustrates the idea of the proof.
\begin{restatable}{lemma}{lemmalefttightmatchingsecond}
\label{lemma:lefttightmatching}
Let $X=\overline{a b} \subset \RR$ be a line segment and let $\tau$ and $\pi$ be curves with
$[a,b] \subseteq [\tau(0), \tau(1)]$ and
$[a,b] \subseteq [\pi(0), \pi(1)]$.
If  $\d_F(X, \tau)=\delta_1$ and  $\d_F(X, \pi)=\delta_2$, then $\d_F(\tau,\pi) \leq \max(\delta_1,\delta_2)$.
\end{restatable}

\begin{restatable}{theorem}{theoremshortcuttingsignatures}
\label{lemma:deleteanynonsignaturevertex}
Let $\tau$ be a curve with vertices $\tau(t_1),\ldots,\tau(t_m)$, and let $\sigma_{\tau}$ be a $\delta$-signature of $\tau$ with vertices $\tau(t_{s_1}),\ldots,\tau(t_{s_{\ell}})$. Let $\tau'$ be a curve obtained by deleting any subset of vertices of $\tau$ which are not in $\sigma_{\tau}$, i.e. $\tau'=\langle \tau(t_1'),\ldots,\tau(t_{k}') \rangle $, where 
$\{t_{s_1},\ldots t_{s_{\ell}} \} \subseteq \{t_1',\ldots,t_{k}' \} \subseteq \{t_1,\ldots,t_m \}  $. Then $\d_F(\tau,\tau’) \leq \delta$. 
\end{restatable}

\begin{proof}
Consider any two consecutive vertices of $\sigma_{\tau}$ defined by parameters $t_{s_j}<t_{s_{j+1}}\in [0,1]$. We assume that the parametrization of $\tau'$ is chosen such that $\tau(t_{s_j} ) = \tau' (t_{s_j} )$, for any $j \in [\ell]$.
It suffices to show that  \[\d_F\left(\tau'[t_{s_j},t_{s_{j+1}}],\tau[t_{s_j},t_{s_{j+1}}]\right)\leq \delta\] for any $j\in [\ell-1]$, because we can then concatenate partial matchings of $\tau'[t_{s_j},t_{s_{j+1}}]$ with $\tau[t_{s_j},t_{s_{j+1}}]$, for all $j\in[\ell-1]$, and obtain a matching of $\tau$ with $\tau'$. 
By Lemma \ref{lemma:signatures1}, we know that for each $j\in[\ell-1]$,  $\d_F\left(\tau[t_{s_j},t_{s_{j+1}}],\overline{\tau(t_{s_j})\tau(t_{s_{j+1}})}\right)\leq \delta$, since $\overline{\tau(t_{s_j})\tau(t_{s_{j+1}})}$ is a $\delta$-signature of $\tau[t_{s_j},t_{s_{j+1}}]$. Similarly,   \[\d_F\left(\tau'[t_{s_j},t_{s_{j+1}}],\overline{\tau(t_{s_j})\tau(t_{s_{j+1}})}\right)\leq \delta,\] because $\overline{\tau(t_{s_j})\tau(t_{s_{j+1}})}$ is a $\delta$-signature of $\tau'[t_{s_j},t_{s_{j+1}}]$. Then, by Lemma \ref{lemma:lefttightmatching}, \[\d_F\left(\tau'[t_{s_j},t_{s_{j+1}}],\tau[t_{s_j},t_{s_{j+1}}]\right)\leq \delta.\]
  \end{proof}

\begin{figure}[t]
    \centering
    \includegraphics[width=\textwidth, page=2]{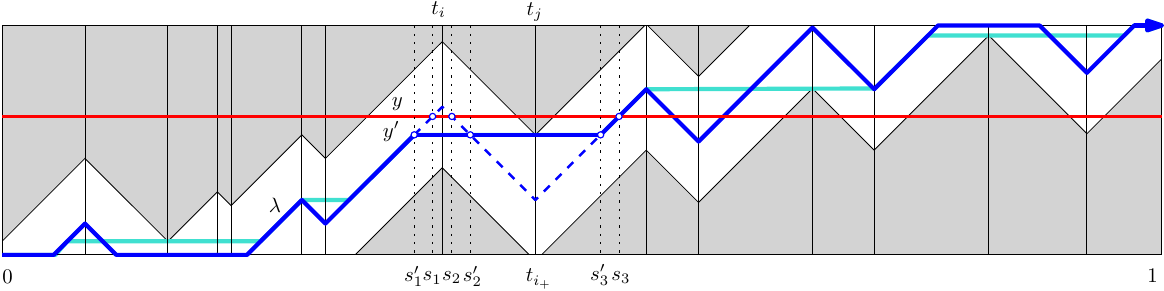}
    \caption{Replacing a section of the path with a horizontal line segment in the proof of Lemma~\ref{lemma:lefttightmatching1}}
    \label{fig:lefttight}
\end{figure}

\subsection{Full Proofs of Section~\ref{section:tightmatchings}}
\label{app:proofs}

\lemmalefttightmatchingfirst*

\begin{proof}
Consider the $\delta$-free space of the two curves $X$ and $\tau$, which is a subset of $[0,1]^2$. We adopt the convention that a point $(x,y) \in [0,1]$ in this diagram corresponds to two points $X(y)$ and $\tau(x)$ (so $X$ corresponds to the vertical axis and $\tau$ corresponds to the horizontal axis). Let $0=x_1\leq \dots\leq x_p=1$ denote the parameter values at vertices of $\tau$. The $\delta$-free space is subdivided into cells $[0,1] \times [x_{i}, x_{i+1}]$.
We call the intersection of the $\delta$-free space with the vertical cell boundary at $x$-coordinate $x_i$ the free space interval at index $i$ and denote it with $[\ell_i,r_i]$.
Consider the $0$-free space inside this diagram, this is the set of points $(x,y) \in [0,1]^2$ with $X(y)=\tau(x)$. 
This set forms a set of paths $\lambda_1,\dots,\lambda_r$, for some $r \in \mathbb{N}$, which is $x$-monotone, since $X$ is a line segment. Therefore, we can parameterize this set by $x$. 
We concatenate any two $\lambda_{i}$ and $\lambda_{i+1}$ by adding a line segment between their endpoints. 
A connecting segment will be a horizontal line, either at $y=0$ or at $y=1$. This can easily be proved by contradiction (assume that $\lambda_i$ ends at $0$ and $\lambda_{i+1}$ starts at $1$, then the section of $\tau$ between those endpoints would have to be disconnected).
In addition, we add line segments to connect $\lambda_1$ to $(0,0)$ and to connect $\lambda_r$ to $(1,1)$. We obtain a connected path $\lambda$ from $(0,0)$ to $(1,1)$, which lies inside the $\delta$-free space, but is not necessarily monotone in $y$. Figure~\ref{fig:lefttight} shows an example.

We now describe how to obtain a $\delta$-tight matching from $\lambda$ by repeatedly replacing sections of $\lambda$ with horizontal line segments, until $\lambda$ is monotone in both parameters, $x$ and $y$.

Assume $\lambda$ is not monotone. Then, there exists a horizontal line that properly intersects $\lambda$ in three different points. Consider a horizontal line at height $y$ with three distinct intersections at $(s_1,y)$, $(s_2,y)$, and $(s_3,y)$, such that
\begin{compactenum}[(i)]
\item the section of $\lambda$ between $s_1$ and $s_2$ lies completely above $y$
\item the section of $\lambda$ between $s_2$ and $s_3$ lies completely below $y$
\end{compactenum}
There exist indices $i$ and $j$, such that $s_1 \leq t_i < t_j \leq s_3$ and such that $t_i$ is minimal and $t_j$ is maximal in this set of indices. 
Let $L$ be the line segment from $(s_1,y)$ to $(s_3,y)$. If $L$ is contained inside the $\delta$-free space, then we replace the corresponding section of $\lambda$ with $L$ and obtain monotonicity of $\lambda$ in the cell(s) $[0,1]\times[x_i,x_j]$.  

Otherwise,  let $i_{-} \in [i,j]$ be the index that maximizes $\ell_{i_{-}}$ and let $i_{+}$ be the index in $[i,j]$ which minimizes $r_{i_{+}}$. (Recall that $[\ell_i,r_i]$ denotes the free space interval at index $i$). It must be that $y \notin [\ell_{i_{-}}, r_{i_{+}}]$, otherwise the line segment $L$ would be contained inside the $\delta$-free space.

Assume $y > r_{i_{+}}$ (the other case is symmetric and handled below). This case is illustrated in Figure~\ref{fig:lefttight}. Let $y'=r_{i_{+}}$ and consider the intersections of $\lambda$ with the horizontal line at $y'$. It must be that there exist intersection points with $s'_1, s'_2, s'_3$ with $s'_1 < s_1 < s'_2 < s_3' < s_3$, such that
\begin{compactenum}[(i)]
\item the section of $\lambda$ between $s'_1$ and $s'_2$ lies completely above $y'$
\item the section of $\lambda$ between $s'_2$ and $s'_3$ lies completely below $y'$
\end{compactenum}
Let $L'$ be the line segment from $(s'_1,y')$ to $(s'_3,y')$.
Since $\d_F(X,\tau)\leq \delta$, it holds that $\ell_j \leq r_{i_{+}}$ for any $j \leq i_{+}$, otherwise there cannot be a monotone path in the $\delta$-free space. 
Therefore, $L'$ is contained in the $\delta$-free space and we can use it to shortcut $\lambda$ and obtain monotonicity of $\lambda$ in the cell(s) $[0,1]\times[x_i,x_j]$.

Otherwise, we have $y < \ell_{i_{-}}$. We handle this case symmetrically. Let $y'=\ell_{i_{-}}$ and consider the intersections of $\lambda$ with the horizontal line at $y'$. It must be that there exist intersection points with $s'_1, s'_2, s'_3$ with $s'_1 < s'_2 < s_1 < s_3' < s_3$, such that 
\begin{compactenum}[(i)]
\item the section of $\lambda$ between $s'_1$ and $s'_2$ lies completely above $y'$
\item the section of $\lambda$ between $s'_2$ and $s'_3$ lies completely below $y'$
\end{compactenum}
Let $L'$ be the line segment from $(s'_1,y')$ to $(s'_3,y')$.
Since $d_F(X,\tau)\leq \delta$, it holds that $\ell_{i_{-}} \leq r_j \leq $ for any $j \geq i_{-}$, otherwise there cannot be a monotone path in the $\delta$-free space. 
Therefore, $L'$ is contained in the $\delta$-free space and we can use it to shortcut $\lambda$ and obtain monotonicity of $\lambda$ in the cell(s) $[0,1]\times[x_i,x_j]$.

With each shortcutting step we obtain monotonicity of the path $\lambda$ in at least one of the cells. Therefore, the process ends after a finite number of steps.

  \end{proof}

\lemmalefttightmatchingsecond*

\begin{proof}
Let $\delta=\max(\delta_1,\delta_2)$.
By Lemma~\ref{lemma:lefttightmatching1} there exists a $\delta$-tight matching from $X$ to $\tau$ and another one from $X$ to $\pi$. We construct a monotone path in the $\delta$-free space of $\tau$ and $\pi$ from these two tight matchings.
In particular, we first specify diagonal segments of the constructed path, which lie in the $0$-free space, and then connect these segments with horizontal, resp., vertical segments.
Let $S \subset [0,1]$ be the finite set of parameter values of $X$, which correspond to the horizontal segments of the tight matching from $X$ to $\pi$. Let $Q \subset [0,1]$ be the finite set of parameters of the horizontal segments of the tight matching from $X$ to $\tau$. Let $y_1 < \dots < y_{r}$ be the sorted list of the values $S \cup Q$ (without multiplicities). 
For any interval $y_i,y_{i+1}$ in this list, there exists a diagonal segment in both tight matchings that covers the entire interval in the $y$-direction. That is, the tight matching matches $X[y_i,y_{i+1}]$ to a subcurve on $\tau$ and a subcurve on $\pi$ that are identical. Let $\tau[t,t']$ and $\pi[p,p']$ be these subcurves. Let $\lambda_{i}=\overline{(t,p)(t',p')}$ be the corresponding diagonal segment of the $\delta$-free space of $\tau$ and $\pi$. Since the two subcurves are identical, $\lambda_i$ is part of the $0$-free space.
We obtain a set of diagonal segments in the $0$-free space, which we intend to connect to piecewise-linear path where every edge is of one of three types:
\begin{inparaenum}[(i)]
\item a diagonal edge contained in the $0$-free space,
\item a horizontal edge,
\item a vertical edge.
\end{inparaenum}
For connecting two diagonal segments $\lambda_i$ and $\lambda_{i+1}$, there are three cases:
\begin{compactenum}[(i)]
 \item $y_{i+1} \in S$ and $y_{i+1} \notin Q$: in this case $
 \lambda_i$ and $\lambda_{i+1}$ can be connected by a horizontal line segment.
 \item $y_{i+1} \notin S$ and $y_{i+1} \in Q$: in this case $
 \lambda_i$ and $\lambda_{i+1}$ can be connected by a vertical line segment.
 \item $y_{i+1} \in S$ and $y_{i+1} \in Q$: in this case $\lambda_i$ and $\lambda_{i+1}$ can be connected by a horizontal line segment followed by a vertical line segment.
\end{compactenum}
From this, we obtain a monotone path in the $\delta$-free space of $\pi$ and $\tau$ from $(0,0)$ to $(1,1)$.
  \end{proof}

\subsection{The data structure}
\paragraph{The data structure} 
The input consists of a set $\Pi$ of $n$ curves in $\RR$, 
and the approximation error $\eps>0$. As before, we assume that the distance threshold is $r:=1$ (otherwise we can uniformly scale the input).
To discretize the query space, we use the regular grid $\GGG_{w}:=\{i\cdot w \mid i\in \ZZ\}$, where $w := \eps/2$. Let $\HHH$ be a dictionary which is initially empty. 
For each input curve $\pi \in \Pi$, with vertices $\VVV(\pi)=v_1,\ldots,v_m$, we set $r_i=[v_i-4-w,v_i+4+w]$, for $i\in [m]$, and 
we compute a set $\CCC':=\CCC'(\pi)$ which contains all curves  
with at most $k$ vertices such that each vertex belongs to some $ r_i\cap \GGG_w$ and the vertices are ordered in the order of $i$. More formally, 
\[
\CCC'=\{\langle u_1,\ldots,u_{\ell}  \rangle \mid \ell \leq k, \exists i_1,\ldots,i_{\ell} \text{ s.t. } i_1\leq \dots \leq i_{\ell}  \text{ and } \forall j\in [\ell]~ u_j \in r_{i_j}\cap \GGG_w \}. 
\]
 Next, we filter $\CCC'$ to obtain the set $\CCC(\pi):=\{\sigma \in \CCC' \mid \d_F(\sigma,\pi)\leq 1+w\}$.

 We store $\CCC(\pi)$ in $\HHH$ as follows: for each $\sigma \in \CCC(\pi)$, we use as  key   the sequence of its vertices $\VVV(\sigma)$: 
if $\VVV(\sigma)$ is not already stored in $\HHH$, then we insert the pair ($\VVV(\sigma)$,$\pi$)
into $\HHH$.  
The total space required is $\OO\left(n\cdot \max_{\pi \in \Pi} |\CCC(\pi)|\right)$. 

\paragraph{The query algorithm}
    For a query curve $\tau$, the algorithm \texttt{query}($\tau$) first 
    computes the $1$-signature of ${\tau}$, namely ${\sigma}$, and then enumerates all possible curves ${\tau}_{key}$ which are produced from ${\tau}$ by deleting vertices that are not in ${\sigma}$. For each possible ${\tau}_{key}$, 
    we compute $\tilde{\tau}_{key}:=\langle g_w(v_1),\ldots,g_w(v_{\ell}) \rangle$, where for any $x\in \RR$, $g_w(x)$ is the nearest point of $x$ in $\GGG_w$. For each $\tilde{\tau}_{key}$
    we perform  a lookup in $\HHH$, with key $\VVV(\tilde{\tau}_{key})$: if $\VVV(\tilde{\tau}_{key})$ is stored in $\HHH$ then we return the associated curve. If there is no $\tilde{\tau}_{key}$ such that $\VVV(\tilde{\tau}_{key})$ is stored in $\HHH$ then the algorithm  returns “no”.

\subsection{Pseudocode}
\label{sectionappendix:new_pseudocode2}

\algorithmus{\texttt{preprocess}(set of time series $\Pi$, $\eps>0$)}{

\begin{algorithmic}[1]
\Statex \Comment{$k$ is assumed to be a variable with global scope}
     \State Initialize empty dictionary $\HHH$ 
     \State $w\gets \eps/2$ 
     \For {{\bf each} $\pi \in \Pi$}
     \State $\CCC(\pi)\gets$ \texttt{generate\_candidates}($\pi$, $w$)
     \If{ $\CCC(\pi)\neq \emptyset$}
     \For{{\bf each} $\sigma_{\tau}\in \CCC(\pi)$}
      \If{$\VVV(\sigma_{\tau})$ not in $\HHH$}
     \State insert key $\VVV(\sigma_{\tau})$ in $\HHH$, associated with a pointer to $\pi$  
     \EndIf
     \EndFor
     \EndIf
     \EndFor
\end{algorithmic}
}

    \algorithmus{\texttt{generate\_candidates}(time series $\pi$ with $\VVV(\pi)=v_1,\ldots,v_m$, $w>0$)}{
\begin{algorithmic}[1]
\For{ {\bf each} $i=1,\ldots,m$}
\State $ r_i\gets  [v_i-4-w,v_i+4+w] $
\EndFor

\State $\CCC' \gets \emptyset$ 
\For{{\bf each} $j=1,\ldots, m $}
\For{{\bf each} $p \in r_j\cap \GGG_w$}
\State \texttt{generate\_sequences}($\langle p\rangle$, $j$, $w$, $\CCC'$) 
\EndFor
\EndFor
\State $\CCC (\pi) \gets \emptyset$
\For{ {\bf each} $\sigma\in \CCC'$}
\If{$\d_F({\pi},\sigma)\leq 1+w$} 
\State $\CCC(\pi) \gets \CCC(\pi) \cup \{ \sigma\}$
\EndIf
\EndFor
\State {\bf  return} $\CCC (\pi) $
\end{algorithmic}
}

 \algorithmus{\texttt{generate\_sequences}(time series $\sigma$, integer $i$, $w>0$, returned set $\CCC'$)}{
\begin{algorithmic}[1]
\Statex \Comment{Stores in $\CCC'$ all possible time series which begin with $\sigma$, have at most $k$ vertices that belong to  $r_j\cap \GGG_{w}$, for $j=i,\ldots, m$, and appear in them in the order of $j$.}
\State $v_1,\ldots,v_t \gets \VVV(\sigma)$
\If {$|\VVV(\sigma)|\leq k$}
\State $\CCC'\gets \CCC' \cup \{ \sigma\}$
\EndIf
\If{$|\VVV(\sigma)|<k$}
\For{{\bf each} $j=i,\ldots, m $}
\For{{\bf each} $p \in r_j\cap \GGG_w$}
\State $\sigma' \gets \langle v_1,\ldots,v_t,p\rangle $ 
\State \texttt{generate\_sequences}($\sigma'$, $j$,$w$, $\CCC'$)
\EndFor
\EndFor
\EndIf
\end{algorithmic}
}

\algorithmus{\texttt{query}(time series $\tau$)}{
\begin{algorithmic}[1]
\Statex \Comment{$w= \eps/2$ is fixed during preprocessing}
\State  ${\tau}(t_1),\ldots,{\tau}(t_{h}) \gets \VVV({\tau}) $
\State $S_{{\tau}} \gets \{t_1,\ldots,t_{h} \}$ \Comment{the set of parameters of vertices of ${\tau} $}
\State ${\sigma}\gets$ $1$-signature of ${\tau}$
\State $S_{{\sigma}} \gets \{s_1,\ldots,s_{\ell}  \}$ \Comment{the set of parameters of vertices of ${\sigma}$ as in $\tau$}
  \State initialize real variables $w_1,\ldots,w_h$
  \For{\textbf{each} $S'\subseteq S_{{\tau}}\setminus S_{{\sigma}}$}
  \State $j\gets 0$
  \For{\textbf{each} $i=1,\ldots,h$}
  \If{$t_i \in  S' \cup S_{{\sigma}}$}
      \State $j\gets j+1$
      \State $w_j \gets {\tau}(t_i)$
  \EndIf
  \EndFor
  \State ${\tau}_{key} \gets \langle w_1,\ldots,w_{j}\rangle $
  \State $\tilde{{\tau}}_{key} \gets$ snap ${\tau}_{key}$ to $\GGG_{w}$
    \If{ $\exists \pi \in \Pi$,  $ \tilde{\tau}_{key}\in \CCC(\pi)$ } \Comment{lookup in $\HHH$}
    \State{\bf report} $\pi$ \Comment{arbitrary $\pi$ s.t.\ $ \tilde{\tau}_{key}\in \CCC(\pi)$ }
  \EndIf
  \EndFor
  \State {\bf report} “no”.
\end{algorithmic}
}

\subsection{Analysis}
We now prove correctness of the query algorithm.

\begin{lemma} 
\label{lemma:correctnessbruteforcequery}
If \texttt{query}($\tau$) returns an input curve $\pi' \in \Pi$, then $\d_F(\pi',\tau)\leq 2+\eps$. If \texttt{query}($\tau$) returns ``no'', then there is no $\pi \in \Pi$ such that $\d_F(\pi,\tau) \leq 1$. 
\end{lemma}
\begin{proof}
For each $i\in [m]$, 
${r}_i:=[\pi(p_i)-4-w,\pi(p_i)+4+w]$, 
$r_i':=[\pi(p_i)-4,\pi(p_i)+4]$, 
where $w=\eps/2$ is the side-length of the grid $\GGG_w$. 
Let ${\sigma}$ be an $1$-signature of ${\tau}$ and let $\tilde{\sigma}$ be the curve obtained by snapping the vertices of $\sigma$ to the grid $G_w$, with $\VVV(\tilde{\sigma})=u_1,\ldots u_{\ell'}$. 
The query algorithm \texttt{query}($\tau$) enumerates all possible curves $\tau_{key}$ which are obtained by deleting any vertices from ${\tau}$ 
which are not vertices of $\sigma$. Let $T_{key}$ be the set of all curves $\tau_{key}$ that are considered by \texttt{query}($\tau$). 
For each curve $\tau_{key}\in T_{key}$, let $ \tilde{\tau}_{key}$ be the curve obtained by snapping the vertices of ${\tau}_{key}$ to the grid $\GGG_w$.  

We first show that if there exists a curve $\tau_{key}\in T_{key}$ and a curve $\pi \in \Pi$  such that $(\VVV(\tilde{\tau}_{key}),\pi)$ is stored in $\HHH$, then $\d_F(\pi,\tau)\leq 2+\eps$. By Lemma~\ref{lemma:deleteanynonsignaturevertex},  any curve $\tau_{key}\in T_{key}$ satisfies $\d_F({\tau},\tau_{key})\leq 1$ and by the triangle inequality $\d_F({\tau},\tilde{\tau}_{key})\leq 1+w$. Since    $(\VVV(\tilde{\tau}_{key}),\pi)$ is stored in $\HHH$, we have that $\tilde{\tau}_{key}\in \CCC(\pi) \implies \d_F(\pi,\tilde{\tau}_{key} )\leq 1+w$. By the triangle inequality
\[
\d_F(\pi,{\tau})\leq \d_F(\pi,\tilde{\tau}_{key})+\d_F(\tilde{\tau}_{key},{\tau}) \leq 2+2w. 
\]

We now show that if there exists $\pi \in \Pi$ such that $\d_F(\pi,\tau)\leq 1$ then there exists ${\tau}_{key}^{\ast}\in T_{key}$ such that the key  $\VVV(\tilde{\tau}_{key}^{\ast})$ is stored in $\HHH$, where $\tilde{\tau}_{key}^{\ast}$ is the curve obtained by snapping the vertices of $\tau_{key}^{\ast}$ to $\GGG_w$.  
Let ${\tau}_{key}^{\ast}$ be the curve obtained by deleting those vertices from ${\tau}$ which are not  vertices of ${\sigma}$ and do not belong to any range $r_i'$. This curve ${\tau}_{key}^{\ast}$ will be considered by \texttt{query}($\tau$), for $S'$ equal to the set of parameters defining vertices of ${\tau}$ which are not in ${\sigma}$ but are contained in  $\bigcup_{i=1}^m r_i'$. 
By Lemma~\ref{lemma:signatures2}, applied on ranges of radius $1$ centered at the vertices of  ${\sigma}$, there 
exist indices $i_1\leq i_2 \leq \ldots \leq i_{|\VVV({\sigma})|}$ such that for each vertex ${\tau}(s_j)$ of  ${\sigma}$, 
${\tau}(s_j) \in r_{i_j}'$.  
By the triangle inequality, there exist indices $i_1\leq i_2 \leq \ldots \leq i_{|\VVV(\tilde{\sigma})|}$ such that for each vertex $u_j$ of  $\tilde{\sigma}$, $u_j \in r_{i_j}$.
Hence, $\tilde{\tau}_{key}^{\ast} \in \CCC'$, where $\CCC'$ is the preparatory set of candidates computed by  \texttt{generate\_candidates}($\pi$). 
Moreover, 
Theorem~\ref{theorem:deletionofvertices} implies that $\d_F(\pi,{\tau}_{key}^{\ast})\leq 1$, because  ${\tau}_{key}^{\ast}$ obtained by deleting vertices of ${\tau}$ which do not belong to any $r_i$.  Hence,  by the triangle inequality, \[\d_F(\pi,\tilde{\tau}_{key}^{\ast})\leq \d_F(\pi,{\tau}_{key}^{\ast})+\d_F(\tilde{\tau}_{key}^{\ast},{\tau}_{key}^{\ast})\leq 1+w \implies  \tilde{\tau}_{key}^{\ast}\in \CCC(\pi),\] where $\CCC(\pi)$ is the final set of candidates as computed and stored by \texttt{generate\_candidates}($\pi$).  Therefore, 
 $\VVV(\tilde{\tau}_{key}^{\ast})$ is stored in $\HHH$, associated with some curve $\pi' \in \Pi $ which satisfies $\d_F(\pi',\tau) \leq 2+\eps$.
\end{proof}

\thmANNFD*
\begin{proof}
By Lemma~\ref{lemma:correctnessbruteforcequery}, the query algorithm returns a correct answer for the ANN problem with distance threshold $r=1$ and approximation factor $1+\eps$. It remains to analyze the complexity of the data structure. 

The space required for each input curve is upper bounded by the number of candidates computed in the preprocessing phase. 
Indeed, we will show now that $|\CCC'|\leq \OO\left( \frac{m}{k\eps}\right)^k$. Recall that the curves in $|\CCC'|$ have vertices in the ranges $r_i\cap \GGG_{w}$, $i=1,\ldots,m$, where $w=
\eps/2$, and the vertices  respect the order of $i$.  In particular, \texttt{generate\_sequences} adds at most one curve to $\CCC'$ for each possible sequence of vertices in   $r_i\cap \GGG_{w}$, $i=1,\ldots,m$, that satisfy the order of $i$. 
If we fix the choices of  $t_1,\ldots,t_{m}$, where each $t_i$ denotes the number of vertices in $r_i\cap \GGG_{w}$ to be  used in the creation of those curves, we can produce at most $\prod_{i=1}^{m} |r_i\cap \GGG_{w}|^{t_i}$ distinct sequences of vertices of length $\sum_{i=1}^{m} t_i$ and hence at most $\prod_{i=1}^{m} |r_i\cap \GGG_{w}|^{t_i}$ curves of length at most $\sum_{i=1}^{m} t_i$. 
Hence, 
\begin{align*}
    |\CCC'| &\leq
\sum_{\substack{t_1+\ldots +t_{m}=k \\ \forall i:~ t_i\geq 0 \\t_1\geq 1,t_{m}\geq 1 }} \prod_{i=1}^{m} \left(\frac{4}{\eps}+2\right)^{t_i } \\ &\leq
\sum_{\substack{t_1+\ldots +t_{m}=k \\ \forall i:~ t_i\geq 0}} \left(\frac{4}{\eps}+2\right)^{k} \\ &\leq {{k+m-1}\choose{k}} \cdot
\left(\frac{4}{\eps}+2\right)^{k} \\ &=\OO\left(\frac{m}{k \eps}\right)^{k},
\end{align*}
which implies that the total storage is in $n\cdot \OO\left(\frac{m}{k \eps}\right)^{k}$.

For each input curve $\pi$, the time needed to 
compute $\CCC(\pi)$ is at most $\OO(|\CCC'|\cdot k\cdot m)$, because we need to compute the Fr\'echet distance between $\pi$ and any curve of $\CCC'$.  
Recall that we employ perfect hashing for $\HHH$, and snapping a curve costs $\OO(k)$ time assuming that a floor function operation needs $\OO(1)$ time.  Hence, the total expected preprocessing time is $\OO(nm)\cdot \OO\left(\frac{m}{k \eps}\right)^{k}$.

 To bound the query time we need to  upper bound the number of distinct curves $\tau_{key}$ which are computed by \texttt{query}($\tau$) in the worst case. There are at most $2^k$ such sets, and for each one of them, we probe the hashtable in $\OO(k)$ time. 
Hence, the total query time is $\OO(k\cdot 2^k)$. 
  \end{proof}

\section{An $\OO(k)$-ANN data structure with  near-linear space}\label{sec:lsh}

In this section we give the data structure for 
Theorem~\ref{theorem:largeapproximation}. 
The data structure has approximation factor of order $\OO(k)$, but it uses space in $\OO(n\log n+nm)$ and query time in $\OO(k \log n)$. Our main ingredient is a properly-tuned randomly shifted grid: 
\label{ss:randomgrids}
Let $w>0$ be a fixed parameter and $z$ chosen uniformly at random from the set $[0,w]$. The function 
$g_{w,z}(x)=\left\lfloor {w}^{-1}({x-z}) \right\rfloor$
induces a random partition of the line.  

\paragraph{The data structure} 
The input consists of a set $\Pi$ of $n$ curves in $\RR$. As before, we assume that the distance threshold is $r:=1$. Let $w=48k$. 
We build $\notables=O(\log n)$ dictionaries $\HHH_1,\ldots, \HHH_{\notables}$ which are initially empty. For each $i\in [\notables]$, we sample $z_i$ uniformly and independently at random from $[0,w]$. 
For each input curve $\pi \in \Pi$, we compute its  $1$-signature $\sigma_{\pi}$, with vertices $\VVV(\sigma_{\pi})=v_1,\ldots,v_{\ell}$, and  for each $i\in [\notables]$ we compute the curve  $\sigma_{\pi|i}'=\langle g_{w,z_i}(v_1),\ldots, g_{w,z_i}(v_{\ell}) \rangle$. For each $\pi \in \Pi$, such that $|V(\sigma_{\pi})|\leq k$, we use as key in $\HHH_i$  the sequence of  vertices $\VVV(\sigma_{\pi|i}')$: if $\VVV(\sigma_{\pi|i}')$ is not already stored in $\HHH_i$, then we insert the 
pair  $(\VVV(\sigma_{\pi|i}'), \pi)$.

\paragraph{The query algorithm} When presented with a query curve $\tau$, with vertices $u_1,\ldots,u_k$, we compute for each $i\in [\notables]$, the curve 
$\tau_i'=\langle g_{w,z_i}(u_1),\ldots,g_{w,z_i}(u_k) \rangle$. Then, for each $i\in [\notables]$, we perform a lookup in   $\HHH_i$ with the key $\VVV({\tau_i}')$ and return the result: if $\exists i \in [\notables]$ such that  $\VVV({\tau_i}')$ is stored in $\HHH_i$ then we return the curve associated with it.  Otherwise we return “no”. (Recall that $\VVV({\tau_i}')$ only retains the maxima and minima of the sequence $g_{w,z_i}(u_1),\ldots,g_{w,z_i}(u_k)$.) 

Figure~\ref{fig:gridalgorithm} shows an example of how keys are computed, both in the case of input curves and in the case of query curves. 

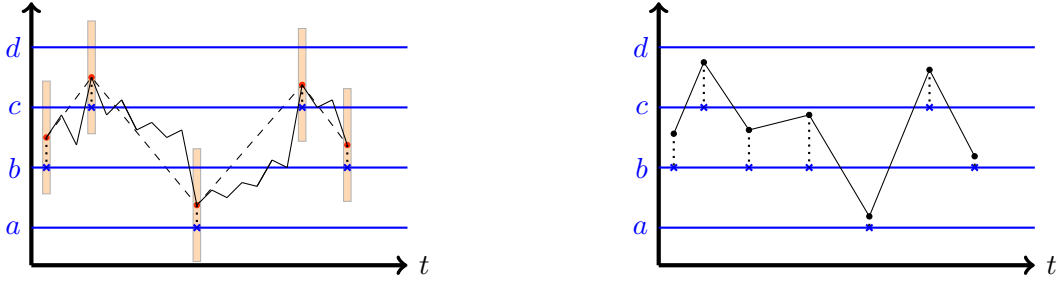
\begin{figure}
    \begin{subfigure}[t]{.5\textwidth}
    \centering
    \begin{tikzpicture}

   \newcommand\T{2}
    \draw[->,ultra thick] (0,-0.5)--(5,-0.5) node[right]{$t$};
    \draw[->,ultra thick] (0,-0.5)--(0,3) node[above]{};
    \coordinate (a) at (0.1*\T,1.2); 
    \coordinate (b) at (0.4*\T,2);
    \coordinate (c) at (1.1*\T,0.3);
    \coordinate (d) at (1.8*\T,1.9);
    \coordinate (e) at (2.1*\T,1.1);
    \filldraw [red] (a) circle (1pt);
    \filldraw [red] (b) circle (1pt);
    \filldraw [red] (c) circle (1pt);
    \filldraw [red] (d) circle (1pt);
    \filldraw [red] (e) circle (1pt);
    \draw[dashed] (a)-- (b)--(c)--(d)--(e);

    \coordinate (AL) at ($(a)-(0,0.8)$);
    \draw[fill=orange,opacity=0.3] ($(a)-(0.05,0.75)$) rectangle ++(0.1,1.5);
    \draw[fill=orange,opacity=0.3] ($(b)-(0.05,0.75)$) rectangle ++(0.1,1.5);
    \draw[fill=orange,opacity=0.3] ($(c)-(0.05,0.75)$) rectangle ++(0.1,1.5);
    \draw[fill=orange,opacity=0.3] ($(d)-(0.05,0.75)$) rectangle ++(0.1,1.5);
    \draw[fill=orange,opacity=0.3] ($(e)-(0.05,0.75)$) rectangle ++(0.1,1.5);
   
   \draw[-,thick, blue]  (0,0)--(5,0) node[pos=0,left]{$a$};
   \draw[-,thick, blue]  (0,0.8)--(5,0.8) node[pos=0,left]{$b$};
    \draw[-,thick , blue ]  (0,1.6)--(5,1.6) node[pos=0,left]{$c$};
    \draw[-,thick , blue ]  (0,2.4)--(5,2.4) node[pos=0,left]{$d$};
    
    \draw (0.1*\T,0.8) node[cross=2pt,rotate=0,blue,thick] {};
    \draw (0.4*\T,1.6) node[cross=2pt,rotate=0,blue,thick] {};
    \draw (1.1*\T,0) node[cross=2pt,rotate=0,blue,thick] {};
    \draw (1.8*\T,1.6) node[cross=2pt,rotate=0,blue,thick] {};
   \draw (2.1*\T,0.8) node[cross=2pt,rotate=0,blue,thick] {};
   \draw[dotted,thick] (a)--(0.1*\T,0.8);
   \draw[dotted,thick] (b)--(0.4*\T,1.6);
    \draw[dotted,thick] (c) -- (1.1*\T,0);
    \draw[dotted,thick] (d) -- (1.8*\T,1.6);
    \draw[dotted,thick] (e) -- (2.1*\T,0.8);

    \draw[-] (0.1*\T,1.2)--(0.2*\T,1.5) node[]{};
    
    \draw[-] (0.2*\T,1.5)--(0.3*\T,1.1) node[]{};
    
    \draw[-] (0.3*\T,1.1)--(0.4*\T,2) node[]{};
    \draw[-] (0.4*\T,2)--(0.5*\T,1.5) node[]{};
    \draw[-] (0.5*\T,1.5)--(0.6*\T,1.7) node[]{};
    \draw[-] (0.6*\T,1.7)--(0.7*\T,1.3) node[]{};
    \draw[-] (0.6*\T,1.7)--(0.7*\T,1.3) node[]{};
        \draw[-] (0.7*\T,1.3)--(0.8*\T,1.4) node[]{};
            \draw[-] (0.8*\T,1.4)--(0.9*\T,1.2) node[]{};
         \draw[-] (0.9*\T,1.2)--(1*\T,1.3) node[]{};
         \draw[-] (1*\T,1.3)--(1.1*\T,0.3) node[]{};
         \draw[-] (1.1*\T,0.3)--(1.2*\T,0.5) node[]{};
         \draw[-] (1.2*\T,0.5)--(1.3*\T,0.4) node[]{}; 
         \draw[-] (1.3*\T,0.4)--(1.4*\T,0.6) node[]{};
             \draw[-] (1.4*\T,0.6)--(1.5*\T,0.55) node[]{};
             \draw[-] (1.5*\T,0.55)--(1.6*\T,0.9) node[]{};                \draw[-] (1.5*\T,0.55)--(1.6*\T,0.9) node[]{};               \draw[-] (1.6*\T,0.9)--(1.7*\T,0.8) node[]{};  
             \draw[-] (1.7*\T,0.8)--(1.8*\T,1.9)node[]{};  
             \draw[-](1.8*\T,1.9)--(1.9*\T,1.6) node[]{};  
            \draw[-](1.9*\T,1.6)--(2*\T,1.7) node[]{};  
             \draw[-](2*\T,1.7)--(2.1*\T,1.1) node[]{};  
            
    \end{tikzpicture}
    \caption{An input time series $\pi$. The red points are vertices of its $\delta$-signature $\sigma_{\pi}$, and the orange rectangles correspond to ranges of radius $\delta$.}
    \end{subfigure}
     \begin{subfigure}[t]{.5\textwidth}
     \centering
    
        \begin{tikzpicture}
   \newcommand\T{2}
   \newcommand\Su{0.2}
   \newcommand\Sd{-0.15}
   \coordinate (a) at (0.1*\T,1.4+\Sd); 
   \coordinate (b) at (0.3*\T,2+\Su);
   \coordinate (c) at (0.6*\T,1.3);
    \coordinate (d) at (1*\T,1.5);
   \coordinate (e) at (1.4*\T,0.3+\Sd);
   \coordinate (f) at (1.8*\T,1.9+\Su);
   \coordinate (g) at (2.1*\T,1.1+\Sd);
     \draw[->,ultra thick] (0,-0.5)--(5,-0.5) node[right]{$t$};
    \draw[->,ultra thick] (0,-0.5)--(0,3) node[above]{};
    \filldraw [black] (a) circle (1pt);
    \filldraw [black] (b) circle (1pt);
    \filldraw [black] (c) circle (1pt);
    \filldraw [black] (d) circle (1pt);
    \filldraw [black] (e) circle (1pt);
    \filldraw [black] (f) circle (1pt);
     \filldraw [black] (g) circle (1pt);
    
   \draw[-,thick, blue]  (0,0)--(5,0) node[pos=0,left]{$a$};
   \draw[-,thick, blue]  (0,0.8)--(5,0.8) node[pos=0,left]{$b$};
    \draw[-,thick , blue ]  (0,1.6)--(5,1.6) node[pos=0,left]{$c$};
    \draw[-,thick , blue ]  (0,2.4)--(5,2.4) node[pos=0,left]{$d$};

     \draw (0.1*\T,0.8) node[cross=2pt,rotate=0,blue,thick] {};
    \draw (0.3*\T,1.6) node[cross=2pt,rotate=0,blue,thick] {};
     \draw ( 0.6*\T,0.8) node[cross=2pt,rotate=0,blue,thick] {};
    \draw ( 1*\T,0.8) node[cross=2pt,rotate=0,blue,thick] {};
    \draw (1.4*\T,0) node[cross=2pt,rotate=0,blue,thick] {};
    \draw (1.8*\T,1.6) node[cross=2pt,rotate=0,blue,thick] {};
   \draw (2.1*\T,0.8) node[cross=2pt,rotate=0,blue,thick] {};
   \draw[dotted,thick] (a)--(0.1*\T,0.8);
   \draw[dotted,thick] (b)--(0.3*\T,1.6);
    \draw[dotted,thick] (c) -- ( 0.6*\T,0.8);
    \draw[dotted,thick] (d) -- ( 1*\T,0.8);
    \draw[dotted,thick] (e) -- (1.4*\T,0);
    \draw[dotted,thick] (f) -- (1.8*\T,1.6);
    \draw[dotted,thick] (g) -- (2.1*\T,0.8);

    
    \draw[-] (a)--(b)--(c)--(d)--(e)--(f)--(g);  
           
    \end{tikzpicture}
    
    \caption{A query time series $\tau$.}
    \end{subfigure}
    \caption{
    Blue lines correspond to grid points. Each vertex is snapped to a grid point. Snapping $\VVV(\sigma_{\pi})$ to the grid produces the sequence $b,c,a,c,b$. The key is $\VVV(\sigma_{\pi}')=\VVV(\langle b,c,a,c,b \rangle) =  b,c,a,c,b $. 
    Snapping $\VVV(\tau)$ to the grid produces the sequence $b,c,b,b,a,c,b$. The key is $\VVV(\tau')=\VVV(\langle b,c,b,b,a,c,b \rangle)=b,c,a,c,b$. The randomly shifted grid has been successfully chosen,
    since $\d_F(\pi,\tau)\leq \delta$ and the two keys are identical.}
    
    \label{fig:gridalgorithm}
\end{figure}

\subsection{Analysis}
We begin with a standard bound on the probability that a randomly shifted grid stabs a given interval.
\begin{lemma}
\label{claim:probset}
Let $X\subseteq \RR$ be a set such that $\diam(X)\leq \Delta$ and $w >0$. Then,
\[ 
 \Pr_z\left[ \exists x\in X~\exists y \in X:~g_{w,z}(x) \neq g_{w,z}(y) \right] \leq \frac{\Delta}{w}.
\]
\end{lemma}
\begin{proof}
 Let $a,b\in \RR$. Then, 
 \[
 \Pr_{z} \left[\left\lfloor \frac{a-z}{w} \right\rfloor \neq \left\lfloor \frac{b-z}{w} \right\rfloor  \right] 
 = \frac{|a-b|}{w}.
 \]
 The claim then follows by setting $a=\min X$, $b=\max X$. 
  \end{proof}

First we focus on any two curves $\pi$, $\tau$ such that $\d_F(\pi,\tau)\leq \delta$. We show that any edge of $\tau$ which is matched to points in the same subcurve $\pi[p_i,p_{i+1}]$, where $p_i$, $p_{i+1}$ are the parameters that correspond to two consecutive signature vertices of $\pi$, and has the opposite direction of that of $\overline{\pi(p_i)\pi(p_{i+1})}$, must be short. This will allow us to argue that any such edge will likely collapse by snapping its vertices to a randomly shifted grid. 

\begin{lemma}\label{claim:weakdirectionpreserving}
Consider any two curves $\pi$, $\tau$ in $\RR$ such that $\d_F(\pi,\tau)\leq \delta$. Let  $\sigma_{\pi}=\pi(p_1),\ldots,\pi(p_{\ell})$ be a $\delta$-signature of $\pi$. Let $0\leq t_1<t_2\leq 1$ be parameters such that each of $\tau(t_1)$, $\tau(t_2)$ is matched with at least one point in $\pi[p_i,p_{i+1}]$, for some $i\in [\ell-1]$, by 
an optimal matching.  Then,
\begin{itemize}
    \item if $\pi(p_i)<\pi(p_{i+1})$ then $\tau(t_2)\geq \tau(t_1)-4\delta$,
    \item if $\pi(p_i)>\pi(p_{i+1})$ then  $\tau\left(t_2\right)\leq \tau\left(t_1\right)+4\delta$.
\end{itemize}
\end{lemma}
\begin{proof}
We prove the case $\pi(p_i)<\pi(p_{i+1})$. The second case is symmetric. Let $\phi$ be an optimal matching between $\pi$ and $\tau$. 
Let $p\in [p_i,p_{i+1}]$ be such that $\pi(p)$ is matched with $\tau(t_1)$ by $\phi$ and let $p'\in [p_i,p_{i+1}]$ be such that $\pi(p')$ is a point  matched with $\tau(t_2)$ by $\phi$. 
By the direction preserving property of $\delta$-signatures, if $\pi(p_i)<\pi(p_{i+1})$ then 
$\pi(p)-\pi(p')\leq 2\delta$. Since $|\pi(p)-\tau(t_{1})|\leq \delta$ and 
$|\pi(p')-\tau(t_2)|\leq \delta$, we have $\tau\left(t_2\right)\geq \tau\left(t_1\right)-4\delta$.
  \end{proof}
Lemma \ref{lemma:signatures2} shows that there exist vertices of $\tau$ which stab the intervals $[\pi(p_i)-\delta,\pi(p_{i})+\delta]$ in the order of $i$. The following claim shows that any subcurve of $\tau$ defined by two vertices of $\tau$ stabbing $[\pi(p_i)-\delta,\pi(p_{i})+\delta]$ and $[\pi(p_{i+1})-\delta,\pi(p_{i+1})+\delta]$ must be entirely contained in the interval $[\min \{\pi(p_i),\pi(p_{i+1}\}-2\delta,\max \{\pi(p_i),\pi(p_{i+1}\}+2\delta]$. In other words, $\tau$ must satisfy a weak analogue of the range property satisfied by signatures. 
\begin{lemma}\label{claim:weakinclusion}
Consider any two curves $\pi$, $\tau$ in $\RR$ such that $\d_F(\pi,\tau)\leq \delta$. Let  $\sigma_{\pi}=\pi(p_1),\ldots,\pi(p_{\ell})$ be a $\delta$-signature of $\pi$. Let $0=t_{j_1} < \dots < t_{j_{\ell}}=1$ be parameters corresponding to vertices of $\tau$ such that $\forall i \in[\ell]$, $|\tau(t_{j_i})-\pi(p_i) | \leq \delta$. Then, for each $i\in \{1,\ldots,\ell-1 \}$, 
\begin{itemize}
    \item if $\pi(p_i)$ is a local minimum, then for any  $x\in \tau[t_{j_{i}},t_{j_{i+1}}]$, it holds $x \geq \pi(p_i)-2\delta$,  
    \item if $\pi(p_i)$ is a local maximum, then for any  $x\in \tau[t_{j_{i}},t_{j_{i+1}}]$, it holds $x \leq \pi(p_i)+2\delta$.
\end{itemize}
\end{lemma}
\begin{proof}
An optimal matching of $\pi$ with $\tau$ matches each $\pi(p_{i})$, $i\in\{2,\ldots,\ell-1\}$, with points in $\tau[t_{j_{i-1}},t_{j_{i+1}}]$. This follows by the monotonicity of an optimal matching, the range property of $\delta$-signatures,  the minimum edge length property of $\delta$-signatures and the triangle inequality. 
 Suppose now that $\pi(p_i)$ is a local minimum. If $i\in \{3,\ldots,\ell-2\}$ then $\pi(p_{i-1})$ is matched with some point in $\tau[t_{j_{i-2}},t_{j_{i}}]$ and $\pi(p_{{i+1}})$ is matched with some point in $\tau[t_{j_{i}},t_{j_{i+2}}]$. If $i=2$ then $\pi(p_{i-1})$ is matched with $\tau(t_{j_{i-1}})$ and $\pi(p_{{i+1}})$ is matched with some point in $\tau[t_{j_{i}},t_{j_{i+2}}]$. If $i=\ell-1$ then  
 $\pi(p_{i-1})$ is matched with some point in $\tau[t_{j_{i-2}},t_{j_{i}}]$ and $\pi(p_{{i+1}})$ is matched with $\tau(t_{j_{i+1}})$.  However, if there exists a point $x$ in $\tau[t_{j_{i-1}},t_{j_{i+1}}]$  such that $x<\pi(p_i)-\delta$, then by the minimum edge length property and the range property of $\delta$-signatures, $x$ cannot be matched with any point in $\pi[p_{i-1},p_{i+1}]$. This implies that the matching is either non-continuous or non-optimal, leading to a contradiction.  
  For $i=1$, by the range property of $\delta$-signatures and and the triangle inequality we have that  for any  $x\in \tau[t_{j_{1}},t_{j_{2}}]$, it holds $x \geq \pi(p_1)-2\delta$. 
 The same arguments can be applied symmetrically when $\pi(p_i)$ is a local maximum.

  \end{proof}
\begin{lemma}
 \label{lemma:randomgrids1}
 Let $\pi$ be a curve in $\RR$ and let 
 $\sigma_{\pi}$ be a $\delta$-signature of $\pi$ with vertices $\pi(p_1),\ldots,\pi(p_{\ell})$. Let $\tau$ be a curve in $\RR$ with vertices $\tau(t_1),\ldots,\tau(t_{k})$. If $\d_F(\pi,\tau)\leq \delta$ then for the two curves $\sigma_\pi'=\langle g_{w,z}(\pi(p_1)),\ldots ,g_{w,z}(\pi(p_{\ell})) \rangle$, 
$\tau'=\langle g_{w,z}(\tau(t_1)),\ldots ,g_{w,z}(\tau(t_{k})) \rangle$ it holds $\VVV(\sigma_\pi')=\VVV(\tau')$ with probability at least $6k\delta/w$, where $z$ is chosen uniformly at random from $[0,w]$. 
\end{lemma} 
\begin{proof}

For each $i\in [\ell]$, we define $r_i:=[\pi(p_i)-\delta,\pi(p_i)-\delta  ]$.
Lemma~\ref{lemma:signatures2} implies that there exist parameters $0=t_{j_1} < \dots < t_{j_{\ell}}=1$ corresponding to vertices of $\tau$ such that $\forall i \in[\ell]$, $\tau(t_{j_i})\in r_i$. 
We first bound the length of edges of any $\tau[t_{j_i},t_{j_{i+1}}]$ which are directed backwards with respect to the direction of $\overline{\pi(p_i),\pi(p_{i+1})}$. We assume that $\pi(p_i)<\pi(p_{i+1})$, since the other case is symmetric. Let $t_1<t_2\in[t_{j_i},t_{j_{i+1}}]$ be two parameters corresponding to two consecutive vertices of $\tau[t_{j_i},t_{j_{i+1}}]$ such that $\tau(t_1)>\tau(t_2)$. Let $\phi$ be an optimal matching of $\pi$ with $\tau$. We consider three cases regarding the position of $\tau(t_1)$:
\begin{enumerate}[label=\roman*)]
    \item if $\tau(t_1)\in[\pi(p_i),\pi(p_{i+1}] \setminus (r_i \cup r_{i+1})$ then $\tau(t_1)$ can only be matched, by $\phi$, with points of $\pi[p_i,p_{i+1}]$ and since $\tau(t_2)<\tau(t_1)$, $\tau(t_2)$ can only be  matched by $\phi$ with points of $\pi[p_i,p_{i+1}]$.  Lemma~\ref{claim:weakdirectionpreserving} implies $|\tau(t_1)-\tau(t_2)|\leq 4\delta$.
    \item If $\tau(t_1)\in r_i$ then by Lemma~\ref{claim:weakinclusion} and the fact that $\tau(t_2)<\tau(t_1)$, we know that  $|\tau(t_1)-\tau(t_2)|\leq 3\delta$.
    \item If $\tau(t_1) \in r_{i+1} \setminus r_i$ then
    \begin{itemize}
        \item if $\tau(t_2)\in r_{i+1}$ then $|\tau(t_1)-\tau(t_2)|\leq 2\delta$. 
        \item if $\tau(t_2)\notin r_{i+1}$ then $\tau(t_2)$ can only be matched, by $\phi$,  with points in $\pi[p_i,p_{i+1}]$. Since $t_1<t_2$, we conclude that $\tau(t_1)$ can also be matched only with points from $\pi[p_i,p_{i+1}]$. Lemma~\ref{claim:weakdirectionpreserving} then implies $|\tau(t_1)-\tau(t_2)|\leq 4\delta$.
    \end{itemize}
\end{enumerate}   
Hence, the length of any edge of any sub-curve  $\tau[t_{j_i},t_{j_{i+1}}]$ which is directed backwards with respect to the direction of $\overline{\pi(p_i),\pi(p_{i+1})}$, has length at most $4\delta$.

For each $i\in [k-1]$, we define $A_i$ as the event that we have $g_{w,z}(\tau(t_i))=g_{w,z}(\tau(t_{i+1}))$ and $I_S\subseteq[k-1]$ denotes the set of indices $i$ such that $|\tau(t_i)-\tau(t_{i+1})|\leq 4\delta$.

For each $i\in [\ell]$, we define $B_i$ as the event that for any two points $x,y\in r_i$ we have $g_{w,z}(x)=g_{w,z}(y)$. We claim that if the event $S=\bigcap_{i\in I_S} A_i \cap \bigcap_{i=1}^{\ell} B_i $ occurs then $\VVV(\sigma_{\pi}')=\VVV(\tau')$.
The event $\bigcap_{i=1}^{\ell} B_i$ directly implies
 that for each $i\in [\ell]$,  $g_{w,z}(\pi(p_i))=g_{w,z}(\tau(t_{j_i}))$. Hence, applying $g_{w,z}(\cdot)$ to the vertices  $\VVV(\tau)$, we obtain a sequence $\VVV(\tau)'$ of the form \[g_{w,z}(\pi(p_1)),\ldots, g_{w,z}(\pi(p_2)),\ldots,g_{w,z}(\pi(p_{\ell})) .\] 
Now, consider any signature edge $\overline{\pi(p_i)\pi(p_{i+1})}$ and suppose that $\pi(p_i) \leq \pi(p_{i+1})$. 
The event $\bigcap_{i\in I_S} A_i $ implies that for any edge $\overline{\tau(t_1)\tau(t_2)}$ of $\tau[t_{j_i},t_{j_{i+1}}]$ with the opposite direction of that of $\overline{\pi(p_i)\pi(p_{i+1})}$, i.e.\ $\tau(t_2)<\tau(t_1)$, we have  $g_{w,z}(\tau(t_1))=g_{w,z}(\tau(t_2)$. Moreover,  $g_{w,z}(\cdot)$ is monotone, which implies that for any two consecutive vertices $\tau(t_1),\tau(t_2)$ in $\tau[t_{j_i},t_{j_{i+1}}]$, regardless of the their direction, we have $g_{w,z}(\tau(t_1))\leq g_{w,z}(\tau(t_2)$. The same arguments apply symmetrically in the case $\pi(p_i)> \pi(p_{i+1})$. In that case any two consecutive vertices $\tau(t_1),\tau(t_2)$ in $\tau[t_{j_i},t_{j_{i+1}}]$, satisfy $g_{w,z}(\tau(t_1))\geq g_{w,z}(\tau(t_2)$. Hence, the sequence $\VVV(\tau)'$ remains monotonic between  $g_{w,z}(\tau(t_{j_i}))=g_{w,z}(\pi(p_{i}))$ and  $g_{w,z}(\tau(t_{j_{i+1}}))=g_{w,z}(\pi(p_{i+1}))$, for any $i\in [\ell]$. This implies that there are no local extrema in $\tau'$ between $g_{w,z}(\tau(t_{j_i}))$ and $g_{w,z}(\tau(t_{j_{i+1}}))$, and hence  the two time series $\tau'$ and $\sigma_{\pi}'$ are identical.

We now upper bound the probability of the complementary event $\overline{S}$:  
\begin{align*}
\Pr\left[ \overline{S}\right]&= 
\Pr\left[ \bigcup_{i\in I_S} \overline{A}_i \cup \bigcup_{i=1}^{\ell} \overline{B}_i \right]\\  &\leq \sum_{i\in I_S}\Pr\left[\overline{A}_i\right] + \sum_{i=1}^{\ell} \Pr\left[\overline{B}_i\right] \\
&   \leq |I_S|\cdot \frac{4\delta}{w} + \ell \cdot \frac{2\delta}{w}\\
&   \leq \frac{6 k \delta}{w},
\end{align*}
where the first two inequalities hold by a union bound, and then we apply Lemma~\ref{claim:probset}.
  \end{proof}

\begin{lemma}
 \label{lemma:randomgrids2}
 Let $\pi$ be a curve in $\RR$ and let 
 $\sigma_{\pi}$ be a $\delta$-signature of $\pi$ with vertices $\pi(p_1),\ldots,\pi(p_{\ell})$. Let $\tau$ be a curve in $\RR$ with vertices $\tau(t_1),\ldots,\tau(t_{k})$. For the two curves $\sigma_\pi'=\langle g_{w,z}(\pi(p_1)),\ldots ,g_{w,z}(\pi(p_{\ell})) \rangle$, 
$\tau'=\langle g_{w,z}(\tau(t_1)),\ldots ,g_{w,z}(\tau(t_{k})) \rangle$, 
if $\VVV(\sigma_\pi')=\VVV(\tau')$ then $\d_F(\pi,\tau)\leq 2w +\delta$. 
\end{lemma}
\begin{proof}
By the triangle inequality, 
\begin{align*}
\d_F(\sigma_\pi,\tau) &\leq \d_F(\sigma_\pi,\sigma_\pi')+ \d_F(\sigma_\pi',\tau) \\ & \leq 
\d_F(\sigma_\pi,\sigma_\pi')+ \d_F(\sigma_\pi',\tau') +\d_F(\tau',\tau)\\ & =
\d_F(\sigma_\pi,\sigma_\pi')+ \d_F(\tau',\tau).
\end{align*}

Notice that $\sigma_\pi'$ and $\tau'$ are  curves resulting by snapping the vertices of   $\sigma_\pi$ and  $\tau$ respectively,  to grid points within distance $w$. 
Hence, $\d_F(\sigma_\pi,\sigma_\pi')\leq w$ and $\d_F(\tau',\tau)\leq w$ which imply   $\d_F(\sigma_\pi,\tau)\leq 2w$. Then by the triangle inequality and Lemma \ref{lemma:signatures1}, 
\[
\d_F(\pi,\tau)\leq \d_F(\pi,\sigma_{\pi})+\d_F(\sigma_{\pi},\tau) \leq 2w +\delta
.\]
  \end{proof}

\thmANNHA*
\begin{proof}
The data structure is described in Section \ref{ss:randomgrids}. We also use notation from that section. 
Each dictionary $\HHH_i$, $i\in [\notables]$, stores for each key a relevant pointer to a curve in $\Pi$. Hence the total storage is in $\OO(nm +n\notables)=\OO(nm + n \log n)$ and the expected preprocessing time is in $\OO(nm \notables)=\OO(nm\log n)$, because we assume perfect hashing. A query costs $\OO(k \notables)=\OO(k\log n)$ time. 

Using Lemmas \ref{lemma:randomgrids1}  for $\delta=r=1$ and $w=12k$, we conclude that for a fixed $i\in [\notables]$, and a query $\tau$ the probability that we get a false negative, meaning that there is a $\pi \in \Pi$ such that $\d_F(\tau,\pi)\leq 1$ but there is no $\pi \in \Pi$  such that $\VVV(\tau_i')=\VVV(\sigma_{\pi|i}')$, is at most   $1/2$. Hence, the probability that we get false negatives in all of the $\notables$ dictionaries is at most $\frac{1}{2^{\notables}}\leq \frac{1}{\poly(n)} $. 
Finally, by Lemma \ref{lemma:randomgrids2}, if there exists  $i\in [\notables]$ such that there is a $\pi \in \Pi$ with $\VVV(\tau_i')=\VVV(\sigma_{\pi|i}')$, then $\d_F(\pi,\tau)\leq 24k+1$. 
  \end{proof}

\section{Distance oracles and asymmetric communication}
\label{sec:lowerbound_reduction}
In this section, we study lower bounds on the cell-probe-complexity of distance oracles for the Fr\'echet distance and the discrete Fr\'echet distance. We focus on the decision version of the problem. In particular, we say a \emph{distance oracle} with input curve $\pi$, threshold $r>0$, and  approximation factor $c>1$, is a data structure which reports as follows: for any query $\tau$, if $\d_{F}(\pi,\tau)\leq r$ then it outputs “yes”, else if $\d_{F}(\pi,\tau)\geq cr$ then it outputs “no” and otherwise both answers are acceptable. This can be viewed as a special case of the $c$-ANN problem. 
To show our lower bounds, we employ a technique first introduced by Miltersen \cite{M94}, which  implies that lower bounds for communication problems can be translated into lower bounds for cell-probe data structures.  
The following communication problem is known as the lopsided (or asymmetric) disjointness problem. 

\begin{definition}[$(k,U)$-Disjointness]
\label{definition:disjointness1}
 Alice receives a set $S$, of size $k$, from a universe
$[U] = \{1 \ldots U\}$, and Bob receives  $T \subset [U]$
of size $m\leq U$. They need to decide whether $T \cap S = \emptyset$. 
\end{definition}

A randomized $[a, b]$-protocol for a communication problem is a protocol in
which Alice sends $a$ bits, Bob sends $b$ bits, and the error
probability is bounded away from $1/2$. 
The following result by P\u{a}tra\c{s}cu gives a lower bound on the randomized asymmetric communication complexity of the $(k,U)$-Disjointness problem.

\begin{theorem}[Theorem 1.4 \cite{P11}]
\label{theorem:disjointnesslb}
Assume Alice receives a set $S$, $|S| = k$ and
Bob receives a set $T$, $|T| = m$, both sets coming from a universe of size $U$, such that $k \leq m \leq U$. 
In any randomized, two-sided error communication protocol deciding disjointness of $S$ and $T$, either Alice sends at least 
$\delta k \log \left(\frac{U}{k} \right)$ bits or Bob sends at least  $ \Omega \left(k \left(\frac{U}{k}\right)^{1-C\cdot\delta} \right)$ bits, for any $\delta > 0$, and $C=1799$.
\end{theorem}

We now define the distance threshold estimation problem (DTEP), where two parties must determine whether two curves are near or far. This is basically the communication version of our data structure problem (for $n=1$).   

\begin{definition}[$(k,U)$-Fr\'echet DTEP]
Given parameters $c\geq 1$, $r>0$, 
Alice receives a curve $\tau$ of complexity $k$ in $\RR^d$, Bob receives a curve $\pi$ of complexity $m\leq U$ in $\RR^d$. If $\d_{F}(\pi,\tau) \leq r$ then they must output “yes”. If $\d_{F}(\pi,\tau) \geq cr$ then they must output “no”. Otherwise, both answers are acceptable. 
\end{definition}

Similarly, we define the $(k,U)$-Discrete Fr\'echet DTEP.

\begin{definition}[$(k,U)$-Discrete Fr\'echet DTEP]
Given parameters $c\geq 1$, $r>0$, 
Alice receives a curve $\tau$ of complexity $k$ in $\RR^d$, Bob receives a curve $\pi$ of complexity $m\leq U$ in $\RR^d$. If $\d_{dF}(\pi,\tau) \leq r$ then they must output “yes”. If $\d_{dF}(\pi,\tau) \geq cr$ then they must output “no”. Otherwise, both answers are acceptable. 
\end{definition}

\subsection{A cell-probe lower bound for the Fr\'echet distance}

Our lower bound of Theorem~\ref{thm:lowerbound} works by reducing the lopsided set disjointness problem to the problem of approximating the Fr\'echet distance of two curves in $\RR$. (A similar reduction appears in \cite{MMR19}, which however works for curves in $\RR^2$.)

First consider an instance of the set disjointness problem: Alice has a set $A=\{\alpha_1,\ldots,\alpha_k\} \subset [U]$ and Bob has a set $B=\{\beta_1,\ldots,\beta_m\}\subset [U]$, where $U$ is the size of the universe.
We now describe our main gadgets which will be used to define one curve of complexity $\OO(k)$ for $A$ and one curve of complexity $\OO(U-m)$ for $B$. 
For each $i \in [U]$:
\begin{itemize}
    \item If $i \in A$ then $x_{2i-1}:=4i+4$, $x_{2i}:=4i$, 
    \item If $i \notin A$ then $x_{2i-1}:=4i$, $x_{2i}:=4i$, 
        \item If $i \in B$ then $y_{2i-1}:=4i$, $y_{2i}:=4i$, 
    \item If $i \notin B$ then $y_{2i-1}:=4i+3$, $y_{2i}:=4i+1$, 
\end{itemize}
We now define $\tilde{x}:=\langle 0,x_1,\ldots,x_{2U}, 4U+5 \rangle$ and $\tilde{y}:=\langle 0 ,y_1,\ldots,y_{2U},4U+5\rangle$. Notice that the number of vertices of $\tilde{x}$ is $2k+2$, and 
the number of vertices of $\tilde{y}$ is 
$2(U-m)+2$, because we only take into account vertices which are local extremes. 
The  arclength of any of $\tilde{x}$, $\tilde{y}$ is at most $12U+2$.

\begin{restatable}{theorem}{frechetDecision}
\label{theorem:frechetdecision}
If $A\cap B=\emptyset$ then $\d_F(\tilde{x},\tilde{y})\leq 1$. 
If $A\cap B\neq \emptyset$ then $\d_F(\tilde{x},\tilde{y})\geq {2}$.
\end{restatable}

\frechetDecision*

\begin{proof}
If there is no $i\in A \cap B$ then there is a monotonic matching which implies $\d_F(\tilde{x},\tilde{y}) \leq 1$. For any $i\in [U]$, let $\tilde{x}_i:=\langle 4i , x_{2i-1},x_{2i},4i+4\rangle$ and $\tilde{y}_i:=\langle 4i , y_{2i-1},y_{2i},4i+4\rangle$. 
To show that, it is sufficient to show that for any $i\in [U]$, $\d_F(\tilde{x}_i, \tilde{y}_i)\leq 1$. If $i \notin A$ and $i \in B$ then the two subcurves are just straight line segments and their distance is $0$. If $i \notin A$ and $i \notin B$ then $\tilde{x}_i$ is a line segment and $\tilde{y}_i$ consists of three line segments forming a zig-zag. The matching works as follows: it first matches the interval  $[4i,4i+2]$ of $\tilde{x}_i$ 
with the interval $[4i,4i+2]$ of $\tilde{y}_i$ by moving in both curves at the same speed, then it stops moving 
in $\tilde{x}_i$, while it moves from $4i+2$ to $y_{2i-1}$ and then to $y_{2i}$ and then to $4i+2$ in $\tilde{y}_i$. The matching continues by moving in  the two remaining subsegments at the same speed. This is a matching that attains $\d_F(\tilde{x}_i, \tilde{y}_i)\leq 1$, because $4i+2$ is within distance $1$ from any of $y_{2i-1},y_{2i}$. Finally if $i\in A$ and $i \notin B$ then the matching works as follows: it first matches $[4i,x_{2i-1}]$ with $[4i,y_{2i-1}]$, then it matches $(x_{2i-1},x_{2i}]$ with $(y_{2i-1},y_{2i}]$, and it finally matches $(x_{2i},4i+4]$ with $(y_{2i},4i+4]$. Since it basically matches pairs of line segments having endpoints at distance at most $1$ from each other, the Fr\'echet distance is again at most $1$. 

Suppose now that there is an $i$ such that $i\in A$ and $i \in B$. Let $v$ be the first appearance of the point $4i+4$ in $
\tilde{x}$, and let $u$ be the second appearance of the point $4i$ in $\tilde{x}$. Assume that $\d_F(\tilde{x},\tilde{y})=\delta <2$. Then, $v$ is matched with some point $z$ in $\tilde{y}$ which lies within distance $\delta$. However,  there is no point in $\tilde{y}$ which lies within distance $\delta$ from $u$, and appears in $\tilde{y}$ after $z$. This implies that $\delta \geq 2$, because the matching required by the definition of the Fr\'echet distance has to be monotonic.  
  \end{proof}
We use a technique of obtaining cell-probe lower bounds first introduced by Miltersen \cite{M94}. For a static data structure problem with input $p \in \PPP$, which computes $f(p,q)$ for any query $q \in \QQQ$, we consider the communication problem, where Alice gets $q \in \QQQ$, Bob
gets $p \in \PPP$, and they must determine $f(q,p)$. If there is a  solution to the data
structure problem with parameters $s, w$ and $t$, then there is a protocol for 
the communication problem, with $2t$ rounds of communication, where Alice
sends $\lceil\log s \rceil$ bits in each of her messages and Bob sends $w$ bits in each of his
messages. The protocol is a simple simulation of the assumed data structure where Alice sends indices to memory cells and Bob responds with the cell content. 
Theorem~\ref{theorem:disjointnesslb}, combined with Theorem~\ref{theorem:frechetdecision},  implies lower bounds for cell-probe Fr\'echet distance oracles. 
\thmDOLB
\begin{proof}
By Theorem \ref{theorem:frechetdecision}, if there exists a randomized $[a,b]$-protocol for the
communication problem, in which, Alice gets any curve $x$ of complexity $2k+2$ and arclength at most $12U+2$, Bob gets any curve $y$ of complexity $2(U-m)+2$ of arclength at most $12U+2$ and they can decide whether $\d_F(x,y) \leq 1$ or $\d_F(x,y) \geq {2}$, then
they can solve the $(k,U)$-Disjointness problem. 

By Theorem \ref{theorem:disjointnesslb}, for  any $\delta>0$, there exists $b_0= \Omega\left(k \left(\frac{U}{k}\right)^{1-1799\delta}\right)$, such that a randomized $[a,b]$-protocol for $(k,U)$-Disjointness, for any $k\leq m\leq U$, requires either $a\geq \delta k \log \left(\frac{U}{k} \right)$ or $b \geq b_0$.  
Hence, for any $\delta>0$, and any $k\leq m \leq U$, if there exists a randomized $[a,b]$-protocol for the $(2k+2,2(U-m)+2)$-Fr\'echet DTEP for any curves of arclength at most $12U+2$, then either $a\geq \delta k \log \left(\frac{
U}{k} \right)$ or $b \geq b_0$.

The simulation argument implies that if there exists a cell-probe data structure with parameters $t$, $w$, $s$ for curves in $\RR$, with query  complexity $2k+2$, and arclength at most $12U+2$, then there exists a randomized  $[2t \log s,2tw]$-protocol for the Fr\'echet DTEP. Hence it should be either that $2t \log s \geq \delta k \log \left(\frac{U}{k} \right)$  or 
$2tw  \geq b_0$. There exists a $w_0= \Omega\left(\frac{k}{2t} \left(\frac{U}{k}\right)^{1-1799\delta}\right)$ such that 
if $w < w_0\leq b_0$, then 
$s\geq   2^{\frac{\delta k \log(U/k)}{2t}}$. 
The theorem is now implied by setting $\delta=\eps/1799$, $L=12U+2$ and rescaling $k\gets 2k+2$. 
  \end{proof}

\subsection{Cell-probe lower bounds for the discrete Fr\'echet distance}

In this section, we focus on distance oracles for the discrete Fr\'echet distance, in the cell-probe model. 
Our reductions use points in a bounded subset of $\RR^d$ requiring $\mathcal{O}(d)$ bits for their description. Next, we define  domains of sequences which satisfy this property.
\begin{definition}[Bounded domain]
We say that a point sequence $P=p_1,\ldots, p_m $ has a \emph{bounded domain} $S\subset \RR^d$ if there exist constants $C>0$, $\lambda>0$ such that for all $i\in [m]$, $p_i \in S$ and each element of $\lambda \cdot p_i$ is an integer lying in $[-C,C]\cap \ZZ$.     
\end{definition}

In the remainder, we reduce $(k,U)$-Disjointness  to $(k,U)$-Discrete Fr\'echet DTEP and conclude with lower bounds for discrete Fr\'echet distance oracles in the cell-probe model.
We consider two cases for $(k,U)$-Discrete Fr\'echet DTEP. First, we assume that points belong to a bounded domain $X\subset \RR^2$ and $|X|=\mathcal{O}(1)$. Second,  we consider the high-dimensional case where points are chosen from some bounded domain $X \subset \RR^{\mathcal{O}(\log m)}$, where $m\leq U$.

\subsection{Constant dimension}

\label{ss:reduction1}
\begin{figure}[t]
\centering
\begin{tikzpicture}[scale=1.9]
        \draw[->,ultra thick] (-2,0)--(2,0) node[right]{$x$};
\draw[->,ultra thick] (0,-1)--(0,1) node[above]{$y$};

   \filldraw[blue] (0,0) circle (1pt) node[anchor=north east] {$s$};
\filldraw[red] (-0.25,0) circle (1pt) node[anchor=south] {$y_1$};
\filldraw[red] (0.25,0) circle (1pt) node[anchor=south] {$y_2$};
\filldraw[blue] (-0.61,0.5) circle (1pt) node[anchor=west] {$\beta_0$};

\filldraw[blue] (-0.61,-0.5) circle (1pt) node[anchor=west] {$\beta_1$};
\filldraw[red] (-1.61,0.5) circle (1pt) node[anchor=west] {$\alpha_0$};
\filldraw[red] (-1.61,-0.5) circle (1pt) node[anchor=west] {$\alpha_1$};
\filldraw[blue] (0.61,0.5) circle (1pt) node[anchor=west] {$\beta_0'$}; 
\filldraw[blue] (0.61,-0.5) circle (1pt) node[anchor=west] {$\beta_1'$};
\filldraw[red] (1.61,0.5) circle (1pt) node[anchor=west] {$\alpha_0'$};
\filldraw[red] (1.61,-0.5) circle (1pt) node[anchor=west] {$\alpha_1'$};
\filldraw[blue] (-0.75,0) circle (1pt) node[anchor=south] {$w$};
\filldraw[blue] (0.75,0) circle (1pt) node[anchor=south] {$w'$};
\filldraw[blue] (-0.25,-1) circle (1pt) node[anchor=south] {$x_1$};
\filldraw[blue] (0.25,-1) circle (1pt) node[anchor=south] {$x_2$};
\filldraw[fill=blue, opacity=0.2, draw=blue] (-0.75,0) circle (1cm);
\filldraw[fill=blue, opacity=0.2, draw=blue] (0.75,0) circle (1cm);
\filldraw[fill=red, opacity=0.2, draw=red] (-0.25,0) circle (1cm);
\filldraw[fill=red, opacity=0.2, draw=red] (0.25,0) circle (1cm);
  \end{tikzpicture}
 \caption{Points used in our gadgets. The blue disks of radius $1$ are centered at $w$ and $w'$ and they cover $\alpha_0,\alpha_1,y_1$ and $\alpha_0',\alpha_1',y_2$ respectively. The red disk of radius $1$ centered at $y_1$ covers $\beta_0,\beta_1,w,w',s, x_1$ and the red disk of radius $1$ centered at $y_2$ covers $\beta_0',\beta_1',w,w',s, x_2$.} 
 \label{figure:gadgets}
\end{figure}
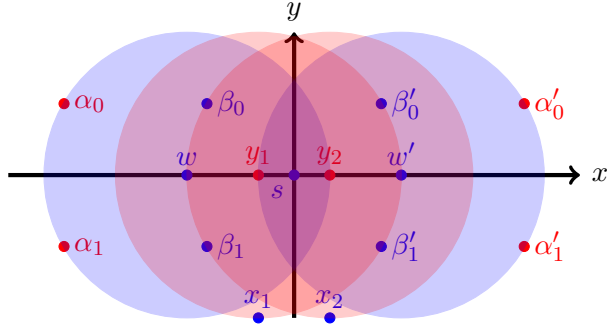

We want to construct point sequences, one for each input set of Alice and Bob, such that there exists a common element in Alice's and Bob's input sets, if and only if the discrete Fr\'echet distance of the two sequences is less or equal than a given threshold. Our reduction takes some of its main ideas from~\cite{BM16}.
Our gadgets use the following points (see Fig.~\ref{figure:gadgets}): 
\[
\alpha_0=\left(-1.61,0.5\right)
,
\alpha_1=\left(-1.61,-0.5\right)
\alpha_0'=\left(1.61,0.5\right)
,
\alpha_1'=\left(1.61,-0.5\right),\]
\[
\beta_0=\left(-0.61,0.5\right)
,
\beta_1=\left(-0.61,-0.5\right)
,
\beta_0'=\left(0.61,0.5\right)
,
\beta_1'=\left(0.61,-0.5\right),\]
\[s=(0,0), 
w=\left(-0.75,0\right),
w'=\left(0.75,0\right),\]
\[
y_1=\left(-0.25,0\right),y_2=\left(0.25,0\right),
x_1=\left(-0.25,-1 \right), 
x_2=\left( 0.25,-1\right).
\]
 Let $D=\lceil \log U \rceil$, where $U$ is the size of the universe in the $(k,U)$-Disjointness instance. We further assume that $D$ is even for convenience. We treat elements of the universe as binary vectors: Alice's set corresponds to a set $\{a_1,\ldots,a_k \}$, where each $a_i \in \{0,1\}^D$, and Bob's set corresponds to a set $\{b_1,\ldots,b_k \}$, where each $b_i \in \{0,1\}^D$. 
For each vector $a_i\in \{0,1\}^D$ we have a gadget $A_i$ which is a sequence of points constructed as follows: for each odd coordinate $j$ we either put 
$\alpha_0$ or $\alpha_1$ depending on whether $(a_i)_j$ is $0$ or $1$ and for each even coordinate $j$ we either put 
$\alpha_0'$ or $\alpha_1'$ depending on whether $(a_i)_j$ is $0$ or $1$. For example, for the vector $(0,1,0,0)$ (assuming that it belongs to Alice) we create $a_0,a_1',a_0,a_0'$. Similarly for each vector $b_i$ we have a gadget $B_i$ which is a sequence of points constructed as follows: for each odd coordinate $j$ we either put 
$\beta_0$ or $\beta_1$ depending on whether $(b_i)_j$ is $0$ or $1$ and for each even coordinate $j$ we either put 
$\beta_0'$ or $\beta_1'$ depending on whether $(b_i)_j$ is $0$ or $1$.   Given two sequences $P=p_1,\ldots,p_m$ and $Q=q_1,\ldots,q_m$, we say that a traversal $T=(i_1,j_1),\ldots,(i_m,j_m)$ is \emph{parallel} if for all $k=1,\ldots,m$ we have $i_k=j_k=k$. 
\begin{lemma}
\label{lemma:vectorgadgets}
 Let $a_i,b_j \in \{0,1\}^D$. If $a_i=b_j$ then $\d_{dF}(A_i,B_j) \leq 1$. If $a_i\neq b_j$ then $\d_{dF}(A_i,B_j) \geq \sqrt{2}$. Moreover, for any non-parallel traversal $T$, we have $\d_T(A_i,B_j) \geq 2$. 
\end{lemma}
\begin{proof}
 If $a_i=b_j$ then the parallel traversal gives $\d_{dF}(A_i,B_j) \leq 1$.  If $a_i\neq b_j$ then  $\d_{dF}(A_i,B_j) \geq \sqrt{2}$. To see that notice that  $\|\beta_0 - \alpha_1\|_2=\|\beta_1 - \alpha_0\|_2=\|\beta_0' - \alpha_1'\|_2=\|\beta_1' - \alpha_1'\|_2=\sqrt{2}$. Furthermore, for each $z,w\in \{0,1\}$ $\|a_z - b_w'\|_2 > 2$ and $\|a_z'\ - b_w\|_2 > 2$.
  \end{proof}
We  
define $W= \bigcirc_{i=1}^{D m/2} (w \circ w')$. Given $a_1,\ldots,a_k$ and $b_1,\ldots, b_m$, we construct two point sequences as follows:
\[
P=W \circ x_1 \circ \bigcirc_{i=1}^m (s \circ B_i) \circ s \circ x_2 \circ W,
\] 
\[
Q= \bigcirc_{i=1}^k (y_1 \circ A_i \circ y_2) .
\]

\begin{lemma}
\label{lemma:isnotempty}
Let $a_1,\ldots,a_k \in \{0,1\}^D$ and $b_1,\ldots, b_m \in \{0,1\}^D$. If there exist $i,j$ such that $a_i = b_j$ then $\d_{dF}(P,Q) \leq 1$. 
\end{lemma}
\begin{proof}
We assume that there exist $i^{*}\in [k]$, $j^{*} \in [m]$ such that $a_{i^{*}} = b_{j^{*}}$. We describe one traversal $T$ which achieves $\d_{T}(P,Q) \leq 1$ and hence $\d_{dF}(P,Q)\leq 1$. 
 \begin{enumerate}
     \item  The first $2D (i^{*}-1)$ points of $W$ are matched with the first $(i^{*}-1)(D+2)$ points of $q$. In particular, for each $i = 1,\ldots, i^{*}-1:$
     \begin{inparaenum}[(i)]
         \item $w$ is matched with $y_1$,
         \item   $T$ proceeds in parallel for  $\bigcirc_{j=1}^{D/2} (w \circ w')$ and $A_i$,
         \item $w'$ is matched with $y_2$.
     \end{inparaenum}
     \item $T$ remains in $y_1$ and it matches it with the rest of $W$. Then, $x_1$ is matched with $y_1$.
     \item  $y_1$ is matched with all points in  $\bigcirc_{j=1}^{j^{*}-1} (s \circ B_i)$.
     \item $T$ proceeds in parallel for  $A_{i^{*}}$ and $B_{j^{*}}$.
     \item $T$ remains in $y_2$ and proceeds only in $p$ until it reaches $W$. 
     \item The first $2D(m-i^{*})$ points of $W$ are matched with the rest of $Q$ as in step 1.
     \item $T$ remains in $y_2$ (the last point of $q$) and it proceeds in $P$ until the end.
 \end{enumerate}
Points $w$, $w'$ are within distance $1$ from any of $y_1$, $y_2$, $\alpha_0$, $\alpha_1$, $\alpha_0'$, $\alpha_1'$. Points $y_1$ are within distance $1$ from $x_1$, $s$ and any of $\beta_0$, $\beta_1$, $\beta_0'$, $\beta_1'$. By Lemma \ref{lemma:vectorgadgets}, $\d_{dF}(A_{i^{*}},B_{j^{*}}) \leq 1$. Then $y_2$ is within distance $1$ from $x_2$, $s$ and any of  $\beta_0$, $\beta_1$, $\beta_0'$, $\beta_1'$.
  \end{proof}

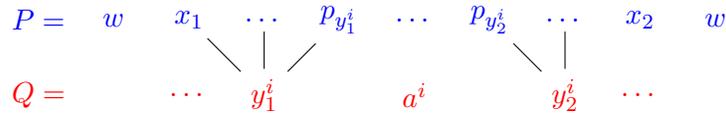
\begin{figure}[h]
\centering
\begin{tikzpicture}[scale=2.0]
   \tikzstyle{q}= [color=blue]
   \tikzstyle{p}= [color=red]
   \node[q] (W1) {$w$};
   \node[q, left of=W1] (q0) {$P=$};
   \node[p, below of =q0] (p0) {$Q=$};
  \node[q, right of=W1] (x1) {$x_1$};
  \node[q, right of=x1] (b1) {$\dots$};
  \node[q, right of=b1] (p1) {$p_{y_1^i}$};
  \node[q, right of=p1] (b2) {$\dots$};
  \node[q, right of=b2] (p2) {$p_{y_2^i}$};
  \node[q, right of=p2] (b3) {$\dots$};
  \node[q, right of=b3] (x2) {$x_2$};
  \node[q, right of=x2] (W2) {$w$};
  \node[p, below of=x1] (a1) {$\dots$};
  \node[p, below of=b1] (y1) {$y_1^i$};
  \node[p, below of=b2] (bi) {$a^i$};
  \node[p, below of=b3] (y2) {$y_2^i$};
  \node[p, right of=y2] (a2) {$\dots$};
    \draw (y1) -- (x1);  
    \draw (y1) -- (b1);
    \draw (y1) -- (p1);  
    \draw (y2) -- (p2);
    \draw (y2) -- (b3);
  \end{tikzpicture}
 \caption{$x_1$ is matched with $y_1^i$, 
 $p_{y_1^i}$ is the last point in $p$ which is matched with $y_1^i$ and $p_{y_2^i}$ is the first point in $p$ which is matched with $y_2^i$. }
 \label{figure:proofgadget}
\end{figure}

\begin{lemma}
\label{lemma:isempty}
Let $a_1,\ldots,a_k \in \{0,1\}^D$ and $b_1,\ldots, b_m \in \{0,1\}^D$. If there are no $i,j$ such that $a_i =b_j$, then $\d_{dF}(P,Q) \geq 1.11$. 
\end{lemma}
\begin{proof}
Consider some traversal $T$. 
We assume that $T$ matches $x_1$ with some $y_1$ and no other point of $Q$. Likewise, $x_2$ is matched with some $y_2$ and no other point of $Q$. If these assumptions do not hold and $x_1$ or $x_2$ are matched with some other point then $\d_T(P,Q)\geq 1.11$. Furthermore we assume that each $s$ is matched with either a $y_1$ or a $y_2$, because otherwise $d_T(P,Q) \geq 1.68$. 
Now let $y_{1}^{i}$ be the $i$th appearance of point $y_1$ in $Q$ and assume that $x_1$ is matched with it. 
Let $p_{y_1^i}$ be the last point in $P$ which is matched with $y_{1}^{i}$ and let $p_{y_2^i}$ be the first point in $P$ which is matched with $y_{2}^{i}$ (see Fig.~\ref{figure:proofgadget}). 
We consider all cases for   $p_{y_1^i}$:

\begin{itemize} 
    \item If $p_{y_1^i}$ is $x_1$ then the first appearance of $s$ is matched with one of $\alpha_0,\alpha_1, \alpha_0', \alpha_1' $ and hence $d_{dF}(P,Q) \geq 1.68$.
    \item If $p_{y_1^i}$ is the $j$th appearance of $s$ then: 
    \begin{itemize}
        \item If $j=k+1$ then the first point of $A_i$ is matched with either $s$ or $x_2$. Hence, the distance is at least $1.68$. 
        \item  If $j<k+1$, then by our initial assumption that $s$ is always matched with either a $y_1$ or a $y_2$, $p_{y_2^i}$ cannot appear after the $(j+1)$th appearance of $s$. Hence, a subsequence of $B_j$ is compared to $A_i$. By Lemma \ref{lemma:vectorgadgets} this implies that $\d_{dF}(P,Q) \geq \sqrt{2}$. 
    \end{itemize}
   
    \item If $p_{y_1^i}$ is a point of some gadget $B_j$ then the same reasoning implies that a subsequence of $B_j$ is compared to $A_i$. By Lemma \ref{lemma:vectorgadgets} this implies that $\d_{dF}(P,Q) \geq \sqrt{2}$.
    \item If $p_{y_1^i} \in \{w,w'\}$ then this means that $x_2$ is matched with $y_1^i$ because of monotonicity of the matching, but then the distance is at least $1.11$.  
\end{itemize}
We conclude that if there are no $i,j$ such that $a_i =b_j$, then $\d_{dF}(P,Q) \geq 1.11$. 
  \end{proof}

\begin{theorem}
\label{theorem:boundedplanarfrechet}
Suppose that there exists a randomized $[a, b]$-protocol for the discrete Fr\'echet DTEP with approximation factor  $c<1.11$ where Alice receives a sequence of $k(2+ \lceil\log(U) \rceil)$ points in $X\subset \RR^{2}$ and Bob receives  a sequence of $3 \lceil\log(U)\rceil m+m +3$ points in $X$, where $X$ is a bounded domain and $|X|\leq 15$. Then there exists a randomized $[a,b]$-protocol for the $(k,U)$-Disjointness problem in a universe $[U]$, where Alice receives a set $S \subset [U]$ of size $k$ and Bob receives a set $T \subseteq [U]$ of size $m$.
\end{theorem}
\begin{proof}
First Alice and Bob convert their inputs to their binary 
representation. Alice uses her binary vectors $a_1,\ldots,a_k$ and constructs a sequence of points $Q= \bigcirc_{i=1}^k (y_1 \circ A_i \circ y_2)$, as described above. Similarly, Bob uses his binary vectors $b_1,\ldots,b_m$ and constructs $P=W \circ x_1 \circ \bigcirc_{i=1}^m (s \circ B_i) \circ s \circ x_2 \circ W$. Then, Alice and Bob run the assumed $[a,b]$-protocol which allows them to determine whether $\d_{dF}(P,Q)\leq 1$ or $\d_{dF}(P,Q)\geq 1.11$. If $\d_{dF}(P,Q)\leq 1$ then the answer to the $(k,U)$-Disjointness instance is “yes” and if $\d_{dF}(P,Q)\geq 1.11$ then the answer is “no”. Lemmas \ref{lemma:isnotempty} and \ref{lemma:isempty} imply that in either case  the answer is correct. 
  \end{proof}

\begin{restatable}{theorem}{theoremcellprobeld}
\label{theorem:cplb1}
Consider any  discrete Fr\'echet distance oracle in the cell-probe model 
which supports point sequences from bounded domains in $\RR^2$, as follows:  for any $k\leq m\leq U$, it stores any  point sequence of length $m$, it  supports  queries of length $k$, and it achieves performance parameters $t$, $w$, $s$, and approximation factor $c<1.11$. 
There exist
\begin{align*}
    w_0 &= \Omega\left(\frac{k}{t \log m}\cdot \left(\frac{U}{k} \right)^{1-\eps}\right), & 
    s_0&=2^{\Omega \left( \frac{ k \cdot  \log(U/k)}{t \log m} \right)},
\end{align*}
such that if $w< w_0$, then $s\geq s_0$, for any constant $\eps>0$. 
\end{restatable}
\begin{proof}
By Theorem \ref{theorem:boundedplanarfrechet}, for sufficiently large $k'=\mathcal{O}(k\log U)$ and $m'=\mathcal{O}(m \log U)$, there exists a  bounded domain $X\subset \RR^2$, for which if there exists a randomized $[a,b]$-protocol for the discrete Fr\'echet DTEP  with approximation factor $c<1.11$,
Alice's input length equal to $k'$, Bob's input length equal to $m'$, 
then there exists a randomized $[a,b]$-protocol for $(k,U)$-Disjointness, where Alice receives a set $S\subseteq [U]$ and Bob receives a set $T\subseteq [U]$ of size $m$, with $k\leq m \leq U$. 

Now consider the following randomized $[a,b]$-protocol. First, Alice and Bob use public random coins to map all elements of $U$ to random bit strings of dimension $D=2 \log(mk)$. By a union bound over at most $mk$ different elements of $U$, distinct elements in $T\cup S$ will be mapped to distinct bit strings with probability at least $1-(mk)^{-1}$. Then, Alice and Bob use the protocol of Theorem \ref{theorem:boundedplanarfrechet} to solve $(k,U)$-Disjointness in a universe of size $m^{\mathcal{O}(1)}$. Hence, for sufficiently large $k''=\Theta(k\log m)$ and $m''=\Theta(m \log m)$, there exists a  bounded domain $X\subset \RR^2$, for which if there exists a randomized $[a,b]$-protocol for the discrete Fr\'echet DTEP  with approximation factor $c<1.11$,
Alice's input length equal to $k''$, Bob's input length equal to $m''$, 
then there exists a randomized $[a,b]$-protocol for $(k,U)$-Disjointness in an arbitrary universe $[U]$, where Alice receives a set $S\subseteq [U]$ and Bob receives a set $T\subseteq [U]$ of size $m$, with $k\leq m$.

By Theorem \ref{theorem:disjointnesslb}, for  any $\delta>0$, a randomized $[a,b]$-protocol for $(k,U)$-Disjointness, for any $m\leq U$, where $U$ is the size of the universe, requires either $a\geq \delta k \log \left(\frac{U}{k} \right)$ or $b \geq b_0$, where $b_0= \Omega\left(k \left(\frac{U}{k}\right)^{1-1799\delta}\right)$. 
Hence, for any $\delta>0$, and any $k$, $m$,  such that $k\leq m$, if there exists a randomized $[a,b]$-protocol for the discrete Fr\'echet DTEP with the above-mentioned input parameters, then either $a\geq \delta k \log \left(\frac{U}{k} \right)$ or $b \geq  b_0$.

The simulation argument implies that if there exists a cell-probe discrete Fr\'echet distance oracle with parameters $t$, $w$, $s$ for point sequences of size $k''$ and $m''$, for points in $D$, then there exists a randomized  $[2t \log s,2tw]$-protocol for the discrete Fr\'echet DTEP. Hence, it should be that either $2t \log s \geq \delta k \log \left(\frac{U}{k} \right)$  or 
$2tw  \geq b_0$. In other words, 
if $w <  b_0$, then 
$s\geq   2^{\frac{\delta k \log(U/k)}{2t}} $. 
Rescaling for $k''=\Theta(k \log m)$ and $m''=\Theta( m \log m)$ implies that there exists 
\[w_0= \Omega\left(\frac{k''}{t \log m}\right)\cdot \left(\frac{U \log m}{k''} \right)^{1-1799\delta}=
\Omega\left(\frac{k''}{t \log m''}\right)\cdot \left(\frac{U }{k''} \right)^{1-1799\delta}
\] 
and 
\[s_0=
2^{\Omega \left( \frac{\delta k''}{t \log m}  \cdot \log \left(\frac{U\log m}{k''} \right)\right)}=2^{\Omega \left( \frac{\delta k''}{t \log m''}  \cdot \log \left(\frac{U}{k''} \right)\right)}
\]
such that 
if $w< w_0$, then $s\geq s_0$. The theorem is now implied by just renaming variables $k''$, $m''$ and setting $\delta=\epsilon/1799$.

  \end{proof}

\subsection{High dimension}
\label{ss:reduction2}
The reduction in the previous section uses point sequences in the plane.
We now describe a second reduction to show a dependency on the ambient dimension $d$ of the point sequences in case $d$ is sufficiently high. 
For all $i\in [U]$, $e_i \in \RR^{U}$ denotes the vector of the standard basis, i.e.\ the vector with all elements equal to $0$ except the $i$-th coordinate which is $1$. We use the following points in $\RR^{U+2}$:
\[
w=(1,1,0,\ldots,0),x_1=(1,-1,0,\ldots,0),
s=(0,0,0,\ldots,0) ,  \tilde{b}_i=(0,0,~~e_i~~),
\]
\[x_2=(-1,1,0,\ldots,0),
y_1=(1,0,0,\ldots,0), \tilde{a}_i=(1,1,~~e_i~~), y_2=(0,1,0,\ldots,0)
\]

Given $S=\{s_1,\ldots,s_k\}$, $T=\{t_1,\ldots,t_m\}$ as in Definition \ref{definition:disjointness1}, we construct the following point sequences:
\[
P=w \circ x_1 \circ \bigcirc_{i=1}^m (s \circ \tilde{b}_{t_i}) \circ s \circ x_2 \circ w,
\] 
\[
Q= \bigcirc_{i=1}^k (y_1 \circ \tilde{a}_{s_i} \circ y_2) .
\]

Notice that $P$ is a point sequence of length $2m+5$ and $Q$ is a point sequence of length $3k$. All points lie in $\RR^{U+2}$. Point $w$ serves as a skipping gadget since it is near to any point of $Q$, and points $s$, $x_1$, $x_2$,$y_1$, $y_2$ are needed for synchronization: $x_1$ is close to $y_1$ but no other point in $Q$,  $x_2$ is close to $y_2$ but no other point in $Q$, and $s$ is close to both $y_1$ and $y_2$ but no other point in $Q$. Our analysis is very similar to the one of Section \ref{ss:reduction2}. A new key component is the use of random projections, and in particular the random projection by Achlioptas \cite{A03} to reduce the dimension. 
\begin{lemma}
\label{lemma:isnotempty2}
If $S\cap T \neq \emptyset$ then $\d_{dF}(P,Q)\leq \sqrt{2}$.
\end{lemma}
\begin{proof}
Let $i^{*}$, $j^{*}$ such that $s_{i^{*}}=t_{j^{*}} \in T\cap S$.  We describe a traversal which achieves distance $\sqrt{2}$: 
\begin{enumerate}
    \item $w$ is matched with all points of $q$ before $y_1 \circ \tilde{a}_{s_{i^{*}}} \circ y_2$
    \item $x_1$ is matched with $y_1$
    \item  $y_1$ is matched with all points of $p$ before $b_{t_{j^{*}}}$
    \item  $\tilde{a}_{s_{i^{*}}}$ is matched with $\tilde{b}_{t_{j^{*}}}$
    \item $y_2$ is matched with the rest of $P$
    \item $w$ is matched with the rest of $Q$
\end{enumerate}
Only the following distances appear in the above matching:
\[
\|w-y_1\|_2,
\|w-y_2\|_2,
\|x_1-y_1\|_2,
\|s-y_1\|_2,
\|\tilde{a}_{s_{i^{*}}}-\tilde{b}_{t_{j^{*}}}\|_2,
\|s-y_2\|_2,
\|x_2-y_2\|_2,
\]
\[
\{\|w-\tilde{a}_{s_i}\|_2 \mid i\neq i^{*}\}, \{\|y_1-\tilde{b}_{t_j}\|_2 \mid j<j^{*}\},\{\|y_2-\tilde{b}_{t_j}\|_2 \mid j>j^{*}\}
\]
and all of them are at most $\sqrt{2}$.
  \end{proof}

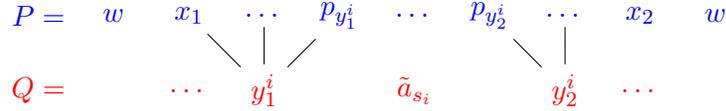
\begin{figure}[h]
\centering
\begin{tikzpicture}[scale=2.3]
   \tikzstyle{q}= [color=blue]
   \tikzstyle{p}= [color=red]
   \node[q] (W1) {$w$};
   \node[q, left of=W1] (q0) {$P=$};
   \node[p, below of =q0] (p0) {$Q=$};
  \node[q, right of=W1] (x1) {$x_1$};
  \node[q, right of=x1] (b1) {$\dots$};
  \node[q, right of=b1] (p1) {$p_{y_1^i}$};
  \node[q, right of=p1] (b2) {$\dots$};
  \node[q, right of=b2] (p2) {$p_{y_2^i}$};
  \node[q, right of=p2] (b3) {$\dots$};
  \node[q, right of=b3] (x2) {$x_2$};
  \node[q, right of=x2] (W2) {$w$};
  \node[p, below of=x1] (a1) {$\dots$};
  \node[p, below of=b1] (y1) {$y_1^i$};
  \node[p, below of=b2] (bi) {$\tilde{a}_{s_i}$};
  \node[p, below of=b3] (y2) {$y_2^i$};
  \node[p, right of=y2] (a2) {$\dots$};
    \draw (y1) -- (x1);  
    \draw (y1) -- (b1);
    \draw (y1) -- (p1);  
    \draw (y2) -- (p2);
    \draw (y2) -- (b3);
  \end{tikzpicture}
 \caption{$x_1$ is matched with $y_1^i$, 
 $p_{y_1^i}$ is the last point in $p$ which is matched with $y_1^i$ and $p_{y_2^i}$ is the first point in $p$ which is matched with $y_2^i$. }
 \label{figure:proofgadget2}
\end{figure}

\begin{lemma}
\label{lemma:isempty2}
If $S\cap T = \emptyset$ then $\d_{dF}(P,Q)\geq \sqrt{3}$.
\end{lemma}
\begin{proof}

Consider the optimal traversal for $P$ and $Q$. 
We assume that $x_1$ is matched with some $y_1$ and no other point of $Q$. Likewise, $x_2$ is matched with some $y_2$ and no other point of $Q$. If these assumptions do not hold and $x_1$ or $x_2$ are matched with some other point then $\d_{dF}(P,Q)\geq \sqrt{5}$. Furthermore we assume that each $s$ is matched with either a $y_1$ or a $y_2$, because otherwise $d_{dF}(P,Q) \geq \sqrt{3}$. 

Now let $y_{1}^{i}$ be the $i$th appearance of point $y_1$ in $Q$ and assume that $x_1$ is matched with it. 
Now let $p_{y_1^i}$ be the last point in $P$ which is matched with $y_{1}^{i}$ and let $p_{y_2^i}$ be the first point in $P$ which is matched with $y_{2}^{i}$ (see Fig.~\ref{figure:proofgadget2}). 

We consider all cases for   $p_{y_1^i}$:

\begin{itemize} 
    \item If $p_{y_1^i}$ is $x_1$ then at least one of the following must happen:
    \begin{itemize}
        \item the first appearance of $s$ is matched with $\tilde{a}_{s_i} $ and hence $d_{dF}(P,Q) \geq \sqrt{3}$,
        \item  $x_1$ is matched with $\tilde{a}_{s_i}$ and hence $d_{dF}(P,Q) \geq \sqrt{3}$. 
    \end{itemize}
    
    \item If $p_{y_1^i}$ is the $j$th appearance of $s$ then: 
    \begin{itemize}
        \item If $j=k+1$ then  $\tilde{a}_{s_i}$ is matched with either $s$ or $x_2$ (or both). Hence, the distance is at least $\sqrt{3}$. 
        \item  If $j<k+1$, then by our initial assumption that $s$ is always matched with either a $y_1$ or a $y_2$, $p_{y_2^i}$ cannot appear after the $(j+1)$th appearance of $s$. Hence, $\tilde{b}_j$ is matched with  $\tilde{a}_i$ (because $s$ is assumed not to be matched with $\tilde{a}_{s_i}$). This implies that $\d_{dF}(P,Q) \geq 2$. 
    \end{itemize}
   \item If $p_{y_1^i}$ is $x_2$ then one of the aforementioned assumptions is not satisfied and hence $d_{dF}(P,Q)\geq \sqrt{3}$. 
    \item If $p_{y_1^i}$ is some point $\tilde{b}_{t_j}$ then $\tilde{a}_{s_i}$ is either matched with $\tilde{b}_{t_j}$ or with $s$. Hence, $\d_{dF}(P,Q) \geq \sqrt{3}$.
    \item If $p_{y_1^i}$ is $w$ then this means that $x_2$ is matched with $y_1^i$ because of monotonicity of the matching, but then the distance is at least $\sqrt{3}$.  
\end{itemize}
We conclude that if $T\cap S=\emptyset$, then $\d_{dF}(P,Q) \geq \sqrt{3}$. 
  \end{proof}

Both point sequences $P$ and $Q$ consist of points in $\{-1,0,1\}^{U+2}$. In order to reduce the dimension, we will use  the following (slighlty rephrased) result by Achlioptas. 
\begin{theorem}[Theorem 1.1 \cite{A03}]
\label{theorem:JL}
Let $P$ be an arbitrary set of $n$ points in $\RR^d$. Given $\eps,\beta>0$, let \[d_0= \frac{4+2\beta}{\eps^2/2-\eps^3/3} \log n.\] 
For integer $d\geq d_0$, let $R$ be a $d' \times d$ random matrix with each $R(i,j)$ being an independent random variable following the uniform distribution in $\{-1,1\}$. With probability at least $1=n^{-\beta}$, for all $u,v\in P$:
\[ \left\|\frac{1}{\sqrt{d'}} Ru -\frac{1}{\sqrt{d'}} Rv \right\|_2 \in (1\pm \eps) \|u-v \|_2.  \]
\end{theorem}

\begin{theorem}
\label{theorem:protocolhd}
Suppose that there exists a randomized $[a, b]$-protocol for the discrete Fr\'echet DTEP with approximation factor  $c<\sqrt{3/2}$ where Alice receives a sequence of $3k$ points in $\left([-3,3]\cap\ZZ\right)^{\Theta( \log m)}$ and Bob receives  a sequence of $2m+5$ points in $\left([-3,3]\cap\ZZ\right)^{\Theta( \log m)}$. Then there exists a randomized $[a,b]$-protocol for the $(k,U)$-Disjointness problem in a universe $[U]$, where Alice receives a set $S \subseteq [U]$ of size $k$ and Bob receives a set $T \subseteq [U]$ of size $m$. 
\end{theorem}
\begin{proof}
Alice constructs a sequence $Q$ of $3k$ points as described above and similarly Bob constructs a sequence $P$ of $2m+5$ points. Let $S$ be the set of all points in $P,Q$. 
Alice and Bob use a source of public random coins to construct the same Johnson Lindenstrauss randomized mapping. In particular, we use Theorem \ref{theorem:JL}.   
Let $R$ be a $d' \times d$ matrix with each element $R(i,j)$ chosen uniformly at random from $\{-1,1\}$ and let $f:~\RR^d \mapsto \RR^{d'}$ be the function which maps any vector $v\in \RR^d$ to $  R v$. 
 Alice and Bob sample $f(\cdot)$ and project their points to dimension $d'=\mathcal{O}(\log (m+k))=\mathcal{O}(\log m)$. With high probability, for any two points $x,y\in S$ we have \[
\| f(x)-f(y)\|_2^2\in \left[0.99\cdot d' \cdot \|x-y\|_2^2 ,1.01 \cdot d' \cdot \|x-y\|_2^2\right].\]  
Each element of vector $f(x)$ is produced by an inner product of a vector of $d$ random signs and a vector of at least $d-3$ zeros and at most $3$ elements from $\{-1,1\}$. 
Hence, $\|f(x)\|_{\infty}\leq 3$ and moreover $f(x) \in \ZZ^{d'}$. Let $f(P)$ and $f(Q)$ be the two point sequences after randomly projecting the points. 
By Lemmas \ref{lemma:isnotempty2} and \ref{lemma:isempty2} we get that if $T\cap S \neq \emptyset$ then $\d_{dF}(f(P),f(Q)) \leq \sqrt{2.02 \cdot d'}$ and if 
$T\cap S = \emptyset$ then $\d_{dF}(f(P),f(Q)) \geq \sqrt{2.97 d'} $. 

Hence Alice and Bob can now use the assumed protocol for computing the discrete Fr\'echet distance and decide whether $T\cap S \neq \emptyset$ or $T\cap S = \emptyset$. 
  \end{proof}

\begin{restatable}{theorem}{theoremcellprobehd}
\label{theorem:cplb2}
There exists $d_0 = \mathcal{O}(\log m)$, such that the following holds. 
Consider any discrete Fr\'echet distance oracle in the cell-probe model which supports point sequences in $\RR^d$,  $d\geq d_0$, as follows:  for any $k\leq m\leq U$, it stores any  sequence of length $m$, it  supports  queries of length $k$, and it achieves performance parameters $t$, $w$, $s$, and approximation factor $c<\sqrt{3/2}$. There exist 
\begin{align*}
w_0&=\Omega \left(\frac{k}{t} \left(\frac{U}{k} \right)^{1-\eps}\right), & 
    s_0&= 2^{\Omega \left(\frac{ k \log(U/k)}{t} \right)}
\end{align*}
such that if $w<w_0$
then 
$s\geq s_0$, for any constant $\eps>0$.

\end{restatable}
\begin{proof}

By Theorem \ref{theorem:protocolhd}, there exists a set $X\subset \RR^{d_0}$, for which if there exists a randomized $[a,b]$-protocol for the discrete Fr\'echet DTEP  with approximation factor $c<\sqrt{3/2}$,
Alice's input length equal to $k'$, Bob's input length equal to $m'$, 
then there exists a randomized $[a,b]$-protocol for $(k,U)$-Disjointness in an arbitrary universe $[U]$, where Alice receives a set $S\subseteq [U]$ and Bob receives a set $T\subseteq [U]$ of size $m$, with $k\leq m \leq U$. By Theorem \ref{theorem:disjointnesslb}, for  any $\delta>0$, a randomized $[a,b]$-protocol for $(k,U)$-Disjointness, for any $m\leq U$, where $U$ is the size of the universe, requires either $a\geq \delta k \log \left(\frac{U}{k} \right)$ or $b \geq b_0$, where $b_0= \Omega\left(k \left(\frac{U}{k}\right)^{1-1799\delta}\right)$.  
Hence, for any $\delta>0$, and any $k$, $m$,  such that $k\leq m$, if there exists a randomized $[a,b]$-protocol for the discrete Fr\'echet DTEP with the abovementioned input parameters, then either $a\geq \delta k \log \left(\frac{U}{k} \right)$ or $b \geq b_0$.

The simulation argument implies that if there exists a cell-probe data structure with parameters $t$, $w$, $s$ for point sequences of size $k$ and $m$, for points in $\RR^d$, then there exists a randomized  $[2t \log s,2tw]$-protocol for the discrete Fr\'echet DTEP. Hence it should be either that $2t \log s \geq \delta k \log \left(\frac{U}{k} \right)$  or 
$2tw  \geq b_0$. There exists a $w_0=\Omega\left(\frac{k}{t} \left(\frac{U}{k}\right)^{1-1799\delta}\right)$ such that 
if $w < w_0$, then 
$s\geq   2^{\frac{\delta k \log(U/k)}{2t}} $. 
The theorem is now implied by just rescaling $\delta=\epsilon/1799$ and substituting for $k'=\Theta(k)$, $m'=\Theta(m)$. 
  \end{proof}

\section{Conclusions}
We have described and analyzed a simple $(5+\eps)$ -ANN data structure. Focusing on improving the approximation factor, while compromising other performance parameters,   we presented a $(2+\eps)$-ANN data structure for time series under the continuous Fr\'echet distance. In doing so, we have presented the new technique of constructing so-called \emph{tight matchings}, which may be of independent interest. In addition, we have also presented a $\OO(k)$-ANN randomized data structure for time series under the Fr\'echet distance, with near-linear space usage and query time in $\OO(k \log n)$. 
We also showed lower bounds in the cell-probe model, which indicate that an approximation better than $2$ cannot be achieved, unless we allow space usage depending on the arclength of the time series or allow superconstant number of probes. 
Our bounds are not tight. In particular, they leave open the possibility of a data structure with approximation factor $(2+\eps)$, with space usage in $n\cdot \OO(\eps^{-1})^k$, and which answers any query using only a constant number of probes.\footnote{In fact, an earlier version of this manuscript claimed such a result, but it contained a flaw.} Moreover, it is possible that even an approximation factor of $(1+\eps)$ can be achieved with  space and query time similar to Theorem~\ref{theorem:firstconstantapprx}.

Apart from these improvements, several open questions remain, we discuss two main research directions:
\begin{compactenum}
    \item Are there data structures with similar guarantees for the ANN problem under the continuous Fr\'echet distance for curves in the plane (or higher dimensions)? Our approach uses signatures, which are tailored to the $1$-dimensional setting. A related concept for curves in higher dimensions is the \emph{curve simplification}. It is an open problem if it is possible to apply simplifications in place of signatures to obtain similar results. 
    \item The lower bounds presented in this paper are only  meaningful  when the number of probes is constant. Can we find lower bounds for the setting that query time is polynomial in $k$ and $m$, and logarithmic in $n$?
\end{compactenum}

One of the aspects that make our results and these open questions interesting is that known generic approaches designed for  general classes of metric spaces cannot be applied. There exist several data structures which operate on general metric spaces with bounded doubling dimension (see e.g. \cite{KL04, HM06, BKL06}). However, the doubling dimension of the metric space defined over the space of time series with the continuous Fr\'echet distance is unbounded \cite{DKS16}. 
Another aspect that makes our problem difficult, is that the Fr\'echet distance does not exhibit a norm structure. In this sense it is very similar to the well-known Hausdorff distance for sets, which is equally challenging from the point of view of data structures (see also the discussion in~\cite{fiHausdoff99,Indyk98}).  We hope that answering the above research questions will lead to new techniques for handling such distance measures.

 \bibliography{jlann}
 \bibliographystyle{splncs04}
 
 \appendix

\newpage
\section{Computational models}
\label{sectionappendix:compmodels}
Our data structures operate in the \emph{real-RAM model}. That is, we assume that the machine can store and access real numbers in constant time and the operations $(+,-,\times,\div, \leq)$ can be performed in constant time on these real numbers. In addition, we assume that the floor function of a real number can be computed in constant time. This model is commonly used in the literature, see for example \cite{H11,C08}. Nonetheless, the use of this computational model is controversial, since it allows all PSPACE and \#P problems to be computed in polynomial time  \cite{BMS85}. 
We stress the fact that, in our algorithms, the floor function is only used in snapping points to a canonical grid. 
In particular, in our data structures, the omission of the floor function (that is, simulating it by the other operations) merely leads to an additional factor in the query time which is bounded by $O(\log(\frac{C}{r}) + \log(\frac{1}{\eps}))$, where $C$ is the largest coordinate of any of the input points and $r$ is the parameter that defines the query radius of the ANN data structure. Moreover, the space and the number of cell probes to the data structure is unaffected by this change.
Our lower bounds hold in the \emph{cell probe model}. In this model of computation we are interested in the number of memory accesses (cell probes) to the data structure which are performed by a query. Given a universe of data and a universe of queries, a cell-probe data structure with performance parameters $s$, $t$, $w$, is a structure which consists of $s$ memory cells, each able to store $w$ bits, and any query can be answered by accessing $t$ memory cells. Note that unlike the real-RAM model, the cell-probe model inherently uses bit-complexity as a measure of space, however the space bounds are usually expressed in terms of the number of words.

\normalsize

\end{document}